\documentclass[twocolumn, 10pt]{article}

\usepackage[margin=0.8in]{geometry}
\usepackage{times}
\usepackage{graphicx}
\usepackage{caption}
\usepackage{booktabs}
\usepackage{array}
\usepackage{multirow}
\usepackage{stfloats}
\usepackage{amsmath, amssymb}
\usepackage{amsthm}
\usepackage{bm}
\usepackage{authblk}

\usepackage{cite}
\usepackage{url}
\usepackage{hyperref}
\hypersetup{
	colorlinks=true,
	linkcolor=blue,
	filecolor=magenta,
	urlcolor=blue,
	citecolor=blue,
}

\usepackage{titlesec}
\titleformat{\section}{\large\bfseries}{\thesection}{0.5em}{}
\titleformat{\subsection}{\normalsize\bfseries}{\thesubsection}{0.5em}{}
\titlespacing*{\section}{0pt}{2.5ex}{1.5ex}
\titlespacing*{\subsection}{0pt}{2ex}{1ex}

\newtheorem{lemma}{Lemma}

\newtheorem{definition}{Definition}
\newtheorem{corollary}{Corollary}
\newtheorem{proposition}{Proposition}

\title{\textbf{Descriptive power and predictive limits of a discrete Hasimoto--DNLS model of protein backbone structure}}

\author[1,2,3]{Yiquan Wang}

\affil[1]{College of Mathematics and System Science, Xinjiang University, Urumqi, Xinjiang, China}
\affil[2]{Xinjiang Key Laboratory of Biological Resources and Genetic Engineering, College of Life Science and Technology, Xinjiang University, Urumqi, Xinjiang, China}
\affil[3]{Shenzhen X-Institute, Shenzhen, China}

\date{}

\begin{document}

\maketitle
\renewcommand*{\thefootnote}{\fnsymbol{footnote}}
\footnotetext[1]{Corresponding author:
	Yiquan Wang (\href{mailto:ethan@stu.xju.edu.cn}{ethan@stu.xju.edu.cn})}
\renewcommand*{\thefootnote}{\arabic{footnote}}

\begin{abstract}
	Determining the three-dimensional structure of a protein from its amino-acid sequence remains a fundamental problem in biophysics. The discrete Frenet geometry of the C$_\alpha$ backbone can be mapped, via a Hasimoto-type transform, onto a complex scalar field $\psi=\kappa\,e^{i\sum\tau}$ satisfying a discrete nonlinear Schr\"odinger equation (DNLS), whose soliton solutions reproduce observed secondary-structure motifs. Whether this mapping, which provides a compact geometric description of folded states, can be extended to a predictive framework for protein folding is an open question. We derive an exact closed-form decomposition of the DNLS effective potential $V_{\text{eff}}=V_{\text{re}}+iV_{\text{im}}$ in terms of curvature ratios and torsion angles, validating the result to machine precision across 856 non-redundant proteins. Our analysis identifies three structural barriers to forward prediction: (i)~$V_{\text{im}}$ encodes chirality through the odd symmetry of $\sin\tau$; its magnitude is ${\sim}31\%$ that of the real part, and neglecting it introduces a $2^N$ degeneracy; (ii)~$V_{\text{re}}$ is determined primarily (${\sim}95\%$) by local geometry, leaving its explicit sequence dependence below ${\sim}5\%$ of the variance; and (iii)~self-consistent field iterations fail to recover native structures (mean RMSD $= 13.1$\,\AA) even with hydrogen-bond terms, yielding torsion correlations indistinguishable from zero. On the constructive side, we show that the residual of the DNLS dispersion relation serves as a geometric order parameter for $\alpha$-helices (ROC AUC $= 0.72$), identifying them as the regions where the backbone most closely approximates an integrable system. These findings indicate that the Hasimoto map functions as a kinematic identity rather than a dynamical governing equation. The obstacles to \textit{ab initio} prediction arise from the purely local, real-potential reduction built upon it, not from the lossless map itself.
\end{abstract}
\vspace{0.5em}

\noindent\textbf{Keywords:} Discrete Hasimoto map; Discrete nonlinear Schr\"odinger equation; Protein backbone geometry; Integrability; Torsion-sign degeneracy

\section{Introduction}

The relationship between amino-acid sequence and three-dimensional
structure is a central problem in molecular biophysics. One geometric approach treats the
backbone as a one-dimensional filament whose shape is captured by its curvature
and torsion. This idea originates in fluid mechanics: Hasimoto~\cite{hasimoto1972soliton} showed that the local induction approximation for a vortex filament can be
exactly transformed, via the complex scalar field
$\psi = \kappa\,e^{i\int\tau\,ds}$, into the cubic nonlinear
Schr\"odinger equation (NLS). Because the NLS is completely integrable,
this transformation converts the geometric evolution of a three-dimensional
curve into a one-dimensional soliton problem with exact analytical
solutions.

An analogous construction applies to the protein backbone, with one difference: the curve is discrete. It is represented as a sequence of points in $\mathbb{R}^3$ whose local shape is
characterized by two scalar fields, the bond angle $\kappa[n]$ and the
torsion angle $\tau[n]$ at each C$_\alpha$ vertex; together these two fields determine
the conformation up to rigid-body motion. A discrete Hasimoto-type transform then encodes these two fields as a
single complex scalar field, and the resulting geometric language has proven highly
successful at \emph{describing} native conformations. Since the native
fold is ultimately encoded in the amino-acid sequence~\cite{anfinsen1973principles},
it is natural to ask whether a \emph{local, integrable reduction} of the backbone geometry
can do more than describe, that is, whether it can \emph{generate} the
three-dimensional fold from sequence-level input alone. This question is distinct and largely uncharacterized.

Niemi and collaborators pursued this analogy systematically for protein backbones, developing a geometric program rooted in gauge field theory and nonlinear dynamics~\cite{danielsson2010gauge, chernodub2010topological, molochkov2017gauge, melnikov2019chern}.
Working with the discrete Frenet frame of the C$_\alpha$ chain, they constructed a generalized discrete nonlinear Schr\"odinger equation (DNLS) whose dark-soliton solutions reproduce the characteristic $(\kappa,\tau)$ profiles of $\alpha$-helices and $\beta$-strands.
Molkenthin, Hu, and Niemi~\cite{molkenthin2011discrete} showed that a two-soliton configuration reproduces the villin headpiece HP35 with RMSD~$=0.72$\,\AA, and that each constituent soliton describes over 7\,000 supersecondary structures in the Protein Data Bank (PDB).
Krokhotin, Niemi, and Peng~\cite{krokhotin2012soliton} constructed a library of 200 soliton parameter sets covering over 90\% of PDB loop structures at sub-\aa ngstr\"om accuracy.
More recently, the framework has been extended to thermal dynamics and structural stability modeling: Begun \textit{et al.}~\cite{begun2021topology} simulated the folding and unfolding of the slipknotted protein AFV3-109 using multi-soliton ans\"atze, and Begun \textit{et al.}~\cite{begun2025local} introduced Arnold's perestroikas to characterize topological phase transitions in myoglobin.
Complementing this topological perspective, Liubimov \textit{et al.}~\cite{liubimov2025modeling} applied the underlying lattice Abelian Higgs model to the same myoglobin structure, demonstrating that native conformations can be stabilized by introducing heterogeneous external fields to mimic environmental interactions.
These studies establish the DNLS and Abelian Higgs frameworks as a compact geometric language for characterizing known protein conformations. The Hasimoto map on which they rest is a rigorous, information-preserving change of variables, a local $U(1)$ gauge transformation of the discrete Frenet frame. These studies constitute a comprehensive treatment of the \emph{inverse problem}: given a known structure, the soliton basis provides an efficient representation and a robust tool for analyzing structural fluctuations. The complementary \emph{forward problem}, generating the native conformation strictly from sequence-level input, is a logically distinct question, and it is the one we take up here.

Concretely, if the DNLS effective potential $V_{\text{eff}}[n]$ could be determined from the amino-acid sequence alone, one could in principle solve the DNLS forward to obtain the native $(\kappa,\tau)$ profile and reconstruct the three-dimensional structure. Whether this is possible turns on whether the DNLS acts as a dynamical governing equation that drives folding or merely as a kinematic identity that records the final state.

Across theoretical biophysics, predictive success has consistently relied on capturing non-local information that the nearest-neighbor structure of the DNLS cannot represent. Energy landscape theory~\cite{dill2012protein} and coarse-grained models like AWSEM~\cite{davtyan2012awsem} achieve accuracy by incorporating explicit long-range contact potentials that bias the free-energy surface. Similarly, from a geometric perspective, tube models~\cite{banavar2007physics, banavar2002geometry} rely on excluded-volume interactions to select secondary structures, while direct coupling analysis (DCA)~\cite{morcos2011direct, cocco2013principal} extracts long-range contacts from evolutionary covariance. Topological features such as knots~\cite{sulkowska2020folding, dabrowski2017topological} further impose global constraints that lie beyond any local curvature description.

The same lesson emerges most clearly from recent deep-learning methods. AlphaFold~2~\cite{jumper2021highly} and AlphaFold~3~\cite{abramson2024accurate} address the prediction problem by inferring a full rigid-body frame (rotation and translation) for every residue, effectively retaining all orientational degrees of freedom. Even single-sequence methods like ESMFold~\cite{rives2021biological, lin2023evolutionary}, driven by protein language models, achieve competitive accuracy by implicitly mapping the sequence into a high-dimensional embedding space that captures non-local structural and dynamic context~\cite{wang2025sequence}. By integrating evolutionary covariance or contextual embeddings, these methods explicitly reconstruct the global information defining the native fold. In contrast, the Hasimoto transform expresses the backbone geometry as a single complex scalar field $\psi$, after which the modeling step reduces the dynamics to a local effective potential. The critical unresolved issue is whether this local real-potential reduction retains sufficient information to determine the three-dimensional structure, or whether it strips away the non-local and chiral information required for folding.

In this work, we derive an exact closed-form decomposition of the DNLS effective potential $V_{\text{eff}} = V_{\text{re}} + i\,V_{\text{im}}$ into explicit functions of the curvature ratios $r^\pm = \kappa[n\pm1]/\kappa[n]$ and torsion angles $\tau[n]$, and verify this decomposition to machine precision ($< 10^{-14}$) on 856 non-redundant proteins spanning all four SCOP structural classes. The decomposition reveals three independent structural barriers to forward prediction. First, the imaginary part $V_{\text{im}}$ encodes the sign of the torsion angle through the odd symmetry of $\sin\tau$, with magnitude roughly 31\% that of the real part; discarding it introduces a $2^N$ chiral degeneracy that grows exponentially with chain length. Second, the real part $V_{\text{re}}$ is 95\% determined by local backbone geometry rather than by the amino-acid sequence, leaving the explicit sequence-to-potential channel too weak to carry the chemical information needed for prediction. Third, self-consistent field iterations driven by hydrophobic and elastic potentials, with and without hydrogen-bond terms, fail uniformly on all 856 proteins (mean RMSD $= 13.1$\,\AA, torsion correlation indistinguishable from zero), and the two settings produce statistically identical outcomes.
Beyond identifying these static and dynamical barriers, we characterize the descriptive-predictive boundary using the language of integrable systems. We prove that within a uniform segment, the DNLS reduces to a constant-coefficient linear recurrence whose transfer matrix has unit determinant, so that the interior of the segment is rigidly fixed by its boundary data. An integrable segment can describe a backbone whose endpoints are given, but it cannot generate an interior conformation from sequence. By classifying the physical interactions governing folding by their parity under torsion inversion and their locality along the chain, we further find that the local, real-potential DNLS occupies a single (even $\times$ local) sector, whereas every driver of tertiary structure (solvent burial, electrostatics, disulfides, and hydrogen bonding) lies in the non-local row. The on-site equations underlying these reductions carry no Lax pair, which provides the structural reason that this description succeeds where \textit{ab initio} generation cannot.
These barriers are structural rather than algorithmic. They arise from
the purely local, real-potential reduction of the DNLS, not from the
Hasimoto map (which is lossless), and are unlikely to be
circumvented by parameter tuning or improved optimization alone. At the same
time, the decomposition yields a constructive result. We show that the
residual of the DNLS dispersion relation serves as a purely geometric
helix detector with ROC AUC $= 0.72$ across all SCOP classes,
identifying $\alpha$-helical segments as the backbone regions where the
DNLS most closely approximates an integrable system. This suggests a
geometric criterion for secondary-structure identification that does not
require hydrogen-bond information.

The remainder of this paper is organized as follows. Section~\ref{sec:frenet}
introduces the discrete Frenet frame and the DNLS formulation.
Section~\ref{sec:decomposition} presents the exact decomposition and its statistical
validation. Section~\ref{sec:piecewise} analyzes piecewise integrability and its
connection to secondary structure. Section~\ref{sec:scf} reports the self-consistent
field tests. Section~\ref{sec:discussion} discusses the physical origin of the barriers
and the constructive applications of the framework.

\section{Discrete Frenet frame and DNLS formulation}\label{sec:frenet}

\subsection{Discrete Frenet geometry of the C$_\alpha$ backbone}

We represent the protein backbone as a polygonal chain of C$_\alpha$
positions $\{\mathbf{r}[n]\}_{n=1}^{N}$ in $\mathbb{R}^3$, where $N$
is the number of residues. The bond vectors are defined as
\begin{equation}
	\mathbf{t}[n] = \mathbf{r}[n+1] - \mathbf{r}[n], \quad n = 1,\dots,N-1.
\end{equation}
For a C$_\alpha$ trace the bond lengths $|\mathbf{t}[n]|$ are
approximately 3.8\,\AA~\cite{engh1991accurate, lovell2003structure}, corresponding to the virtual-bond distance
between consecutive $\alpha$-carbons. We work with the unit tangent
vectors $\hat{\mathbf{t}}[n] = \mathbf{t}[n]/|\mathbf{t}[n]|$.

The discrete Frenet frame at vertex $n$ is constructed from three
successive C$_\alpha$ positions. The bond angle $\kappa[n]$ is defined
through
\begin{equation}
	\cos\kappa[n] = \hat{\mathbf{t}}[n] \cdot \hat{\mathbf{t}}[n-1],
	\quad n = 2,\dots,N-1,
	\label{eq:kappa}
\end{equation}
so that $\kappa[n] \in [0,\pi]$ measures the bending of the chain at
vertex $n$. The torsion angle $\tau[n]$ requires four successive
C$_\alpha$ positions and is defined as the dihedral angle
\begin{equation}
	\begin{split}
		\tau[n] = \text{atan2}\bigl(
		& (\hat{\mathbf{t}}[n-1]\times\hat{\mathbf{t}}[n])
		\times
		(\hat{\mathbf{t}}[n]\times\hat{\mathbf{t}}[n+1])
		\cdot \hat{\mathbf{t}}[n],\\
		& (\hat{\mathbf{t}}[n-1]\times\hat{\mathbf{t}}[n])
		\cdot
		(\hat{\mathbf{t}}[n]\times\hat{\mathbf{t}}[n+1])
		\bigr),
	\end{split}
	\label{eq:tau}
\end{equation}
with $\tau[n] \in (-\pi,\pi]$, defined for $n = 2,\dots,N-2$. The
sign of $\tau$ encodes the local handedness of the backbone: positive
values correspond to right-handed twisting and negative values to
left-handed twisting.

The pair $(\kappa[n],\tau[n])$ constitutes a complete set of internal
coordinates for the discrete curve. Given an initial position
$\mathbf{r}[1]$, an initial tangent $\hat{\mathbf{t}}[1]$, and an
initial normal direction, the full three-dimensional backbone can be
reconstructed from $\{\kappa[n],\tau[n]\}$ by sequential application
of discrete Frenet rotations. This reconstruction is unique up to a
global rigid-body transformation (three translations and three
rotations), so that $\kappa$ and $\tau$ encode all shape information
of the backbone.

The discrete normal and binormal vectors are defined as
\begin{equation}
	\hat{\mathbf{n}}[n]
	= \frac{\hat{\mathbf{t}}[n] - \cos\kappa[n]\;\hat{\mathbf{t}}[n-1]}
	{\sin\kappa[n]},
	\quad
	\hat{\mathbf{b}}[n]
	= \hat{\mathbf{t}}[n-1] \times \hat{\mathbf{n}}[n],
	\label{eq:frame}
\end{equation}
forming an orthonormal triad
$\{\hat{\mathbf{t}}[n-1],\hat{\mathbf{n}}[n],\hat{\mathbf{b}}[n]\}$
at each interior vertex. The torsion angle $\tau[n]$ then governs the
rotation of the normal plane from vertex $n$ to vertex $n+1$ about the
tangent direction.

\subsection{Hasimoto transform and complex scalar field}

Following the construction introduced by Hasimoto for continuous
curves, we define a complex scalar field on the backbone through
\begin{equation}
	\psi[n] = \kappa[n]\,\exp\!\Bigl(i\sum_{k=2}^{n}\tau[k]\Bigr),
	\quad n = 2,\dots,N-2.
	\label{eq:psi}
\end{equation}
The modulus $|\psi[n]| = \kappa[n]$ records the local bending, while
the phase $\arg\psi[n] = \sum_{k=2}^{n}\tau[k]$ accumulates the
torsion along the chain. This transform maps the two real geometric
fields $(\kappa,\tau)$ into a single complex field $\psi$, reducing
the description of backbone geometry to a one-component problem at the
cost of entangling curvature and torsion in the amplitude and phase of
$\psi$. The gauge-theoretic program of Niemi and collaborators~\cite{danielsson2010gauge, chernodub2010topological, molochkov2017gauge, melnikov2019chern} works directly with the curvature and torsion fields of the discrete Frenet frame, treating $\kappa$ and $\tau$ separately. We instead adopt the explicit complex-field form $\psi=\kappa\,e^{i\sum\tau}$ and identify it as a discrete analogue of the Hasimoto transform. This form makes manifest the amplitude-phase entanglement that underlies the present analysis.

The inverse transform recovers the geometric variables as
\begin{equation}
	\kappa[n] = |\psi[n]|, \quad
	\tau[n] = \arg\psi[n] - \arg\psi[n-1].
	\label{eq:inverse}
\end{equation}
Given $\psi[n]$ for all interior vertices, the backbone can therefore
be reconstructed by extracting $(\kappa,\tau)$ via
Eq.~(\ref{eq:inverse}) and applying the discrete Frenet reconstruction.

\subsection{Discrete nonlinear Schr\"odinger equation}

The discrete nonlinear Schr\"odinger (DNLS) equation is a standard model in nonlinear lattice dynamics~\cite{eilbeck1985discrete, kevrekidis2009discrete}; in the version used here for the backbone field $\psi[n]$ it takes the form of a discrete eigenvalue problem
\begin{equation}
	\begin{split}
		&\beta^{+}[n]\bigl(\psi[n+1]-\psi[n]\bigr)
		- \beta^{-}[n]\bigl(\psi[n]-\psi[n-1]\bigr)\\
		&\qquad = V_{\text{eff}}[n]\,\psi[n],
	\end{split}
	\label{eq:dnls}
\end{equation}
where $\beta^{+}[n]$ and $\beta^{-}[n]$ are site-dependent coupling
parameters and $V_{\text{eff}}[n]$ is the effective potential. The
left-hand side is a discrete second difference of $\psi$ with
nonuniform coefficients, analogous to a lattice Laplacian.

The coupling parameters $\beta^{\pm}[n]$ encode the local mechanical
properties of the peptide chain. In the simplest formulation they
depend on the virtual-bond lengths as
\begin{equation}
	\beta^{+}[n] = \frac{1}{|\mathbf{t}[n]|}, \quad
	\beta^{-}[n] = \frac{1}{|\mathbf{t}[n-1]|}.
	\label{eq:beta}
\end{equation}
Because the C$_\alpha$--C$_\alpha$ virtual-bond length varies only
weakly along the chain (standard deviation ${\sim}0.1$\,\AA\ around
the mean of 3.8\,\AA), the coupling parameters are nearly uniform.
This near-uniformity will play a central role in the analysis of
Sec.~\ref{sec:decomposition}, where we show that the sequence dependence of $\beta^{\pm}$
contributes less than 5\% of the variance of $V_{\text{re}}$.

The effective potential is obtained by rearranging
Eq.~(\ref{eq:dnls}):
\begin{equation}
	\begin{split}
		V_{\text{eff}}[n]
		= \frac{1}{\psi[n]}\Bigl[&\beta^{+}[n]\bigl(\psi[n+1]-\psi[n]\bigr)\\
			&- \beta^{-}[n]\bigl(\psi[n]-\psi[n-1]\bigr)\Bigr].
	\end{split}
	\label{eq:veff_def}
\end{equation}
This expression is well defined wherever $\psi[n] \neq 0$, i.e.,
wherever the bond angle $\kappa[n] \neq 0$. For physical protein
backbones $\kappa[n]$ is strictly positive at all interior vertices,
so $V_{\text{eff}}$ is defined throughout the chain.

Because $\psi$ is complex, $V_{\text{eff}}$ is in general complex:
\begin{equation}
	V_{\text{eff}}[n] = V_{\text{re}}[n] + i\,V_{\text{im}}[n].
	\label{eq:veff_split}
\end{equation}
The real and imaginary parts carry distinct geometric content, as we
demonstrate through the decomposition in Sec.~\ref{sec:decomposition}.

Two features of Eq.~(\ref{eq:veff_def}) merit emphasis. First, the
equation is an algebraic identity: for any discrete curve with
$\kappa[n]>0$, one can always compute $\psi$ via
Eq.~(\ref{eq:psi}) and then define $V_{\text{eff}}$ via
Eq.~(\ref{eq:veff_def}). No physical assumption enters this
construction. The potential $V_{\text{eff}}$ is not postulated; it is
read off from the known geometry. Second, the equation becomes a
dynamical equation only when $V_{\text{eff}}$ is specified
independently of the geometry, for example through a physical energy
functional. In the vortex-filament context, the local induction
approximation provides exactly such a specification, and the Hasimoto
transform converts the geometric evolution into the integrable NLS.
For proteins, the question is whether an analogous physical
specification of $V_{\text{eff}}$ exists. The analysis of the
following sections indicates that the answer is negative: the mathematical
form of $V_{\text{eff}}$ on protein backbones presents structural obstacles to its
determination from sequence information alone.

\subsection{Connection to the continuum Hasimoto transform}

In the continuum limit where the lattice spacing $\Delta s \to 0$ and
the coupling becomes uniform ($\beta^{\pm} \to 1/\Delta s$), the
discrete second difference in Eq.~(\ref{eq:dnls}) reduces to
$\partial^2\psi/\partial s^2$, and the DNLS recovers the structure of
the cubic NLS
\begin{equation}
	i\,\frac{\partial\psi}{\partial t}
	= \frac{\partial^2\psi}{\partial s^2}
	+ \tfrac{1}{2}|\psi|^2\psi.
	\label{eq:nls}
\end{equation}
The soliton solutions of Eq.~(\ref{eq:nls}) describe localized
bending excitations that propagate without dispersion along the
filament. In the protein context, the discrete soliton solutions of
Eq.~(\ref{eq:dnls}) have been shown to reproduce the
$(\kappa,\tau)$ profiles of $\alpha$-helices (dark solitons with
$\kappa \approx 1.5$\,rad, $\tau \approx 1.0$\,rad) and
$\beta$-strands (bright solitons with larger $\kappa$ and
$\tau \approx \pm\pi$). The key distinction is that for vortex
filaments Eq.~(\ref{eq:nls}) is both kinematic (relating $\psi$ to
geometry) and dynamic (governing time evolution), whereas for proteins
only the kinematic content survives. Establishing this distinction
quantitatively is the purpose of the remainder of this paper.

\section{Exact discrete decomposition}
\label{sec:decomposition}

The effective potential $V_{\text{eff}}[n]$ defined by Eq.~(\ref{eq:dnls}) encodes the full
geometric content of the backbone in a single complex-valued sequence. In this section we
derive a closed-form decomposition of $V_{\text{eff}}$ into real and imaginary parts, each
expressed entirely in terms of the local curvature ratios and torsion angles. The
decomposition is an algebraic identity: it holds for any discrete space curve, independent
of any physical model or approximation.

\subsection{Decomposition identity}

\begin{proposition}[Decomposition identity]
Let $\psi[n]=\kappa[n]\exp\!\bigl(i\sum_{k=2}^{n}\tau[k]\bigr)$ be the Hasimoto field constructed from the discrete Frenet curvature $\kappa[n]>0$ and torsion $\tau[n]$, and let $V_{\text{eff}}[n]$ be defined through the discrete Schr\"odinger
equation
\begin{equation}
	\beta^{+}_{n}\bigl(\psi[n+1]-\psi[n]\bigr)
	-\beta^{-}_{n}\bigl(\psi[n]-\psi[n-1]\bigr)
	= V_{\text{eff}}[n]\,\psi[n]\,,
	\label{eq:dnls_def}
\end{equation}
where $\beta^{\pm}_{n}$ are the (possibly sequence-dependent) bond stiffness parameters.
Then $V_{\text{eff}}[n]=V_{\text{re}}[n]+i\,V_{\text{im}}[n]$ with
\begin{align}
	V_{\text{re}}[n] &= \beta^{+}_{n}\,r^{+}[n]\cos\tau[n+1]
	+ \beta^{-}_{n}\,r^{-}[n]\cos\tau[n] \notag\\
	&\quad - \bigl(\beta^{+}_{n}+\beta^{-}_{n}\bigr)\,,
	\label{eq:Vre}\\[4pt]
	V_{\text{im}}[n] &= \beta^{+}_{n}\,r^{+}[n]\sin\tau[n+1]
	- \beta^{-}_{n}\,r^{-}[n]\sin\tau[n]\,,
	\label{eq:Vim}
\end{align}
where the curvature ratios are
\begin{equation}
	r^{+}[n] \equiv \frac{\kappa[n+1]}{\kappa[n]}\,,\qquad
	r^{-}[n] \equiv \frac{\kappa[n-1]}{\kappa[n]}\,.
	\label{eq:ratios}
\end{equation}
\end{proposition}

\begin{proof}
We substitute the Hasimoto ansatz into Eq.~(\ref{eq:dnls_def}) and divide
both sides by $\psi[n]\neq 0$, which gives
\begin{equation}
	\frac{V_{\text{eff}}[n]\,\psi[n]}{\psi[n]}
	= \beta^{+}_{n}\!\left(\frac{\psi[n+1]}{\psi[n]}-1\right)
	-\beta^{-}_{n}\!\left(1-\frac{\psi[n-1]}{\psi[n]}\right).
\end{equation}
The ratio of adjacent Hasimoto fields is
\begin{equation}
	\frac{\psi[n+1]}{\psi[n]}
	= \frac{\kappa[n+1]}{\kappa[n]}\,
	\exp\!\bigl(i\,\tau[n+1]\bigr)
	= r^{+}[n]\,e^{i\tau[n+1]}\,,
\end{equation}
and similarly
\begin{equation}
	\frac{\psi[n-1]}{\psi[n]}
	= \frac{\kappa[n-1]}{\kappa[n]}\,
	\exp\!\bigl(-i\,\tau[n]\bigr)
	= r^{-}[n]\,e^{-i\tau[n]}\,.
\end{equation}
Collecting terms,
\begin{equation}
	V_{\text{eff}}[n]
	= \beta^{+}_{n}\bigl(r^{+}[n]\,e^{i\tau[n+1]}-1\bigr)
	-\beta^{-}_{n}\bigl(1-r^{-}[n]\,e^{-i\tau[n]}\bigr)\,.
\end{equation}
Separating real and imaginary parts via Euler's formula yields
Eqs.~(\ref{eq:Vre})--(\ref{eq:Vim}) directly.
\end{proof}
The derivation is algebraically direct; the value of stating it explicitly lies in
disentangling the roles of $r^{\pm}$ and $\tau$ that are obscured in the complex field
$\psi$, and in enabling the systematic tests of Secs.~\ref{sec:decomposition}--\ref{sec:scf}.

Two further features of this decomposition merit emphasis. First, the result is exact: no
continuum limit, Taylor expansion, or slowly varying envelope approximation has been
invoked. Second, the decomposition holds for arbitrary bond parameters $\beta^{\pm}_{n}$
and for any discrete curve with $\kappa[n]>0$; it is a kinematic identity rather than a
dynamical equation.

\subsection{Numerical verification}

We verify the decomposition on eight representative proteins spanning the four SCOP~\cite{murzin1995scop, andreeva2020scop}
structural classes (Table~\ref{tab:decomp}). For each protein, we compute $V_{\text{eff}}[n]$
in two independent ways: (i) directly from the definition Eq.~(\ref{eq:dnls_def}) using the
Hasimoto field $\psi[n]$, and (ii) from the analytic expressions
Eqs.~(\ref{eq:Vre})--(\ref{eq:Vim}) using $\kappa[n]$ and $\tau[n]$. The maximum absolute
difference over all residues is reported in the column labeled $\epsilon$ in
Table~\ref{tab:decomp}.

\begin{table*}[t]
	\centering
	\caption{Numerical verification of the exact decomposition on eight representative proteins
		(two per SCOP class). $N$: number of C$_\alpha$ atoms; $\epsilon$: maximum absolute error
		between the two independent computations of $V_{\text{eff}}$;
		$\langle|V_{\text{im}}|/|V_{\text{re}}|\rangle$: mean ratio of imaginary to real potential
		magnitude; $\rho_{\text{geom}}$: Spearman correlation between $V_{\text{re}}$ computed with
		physical $\beta(s)$ and with uniform $\beta=1$; $\delta_H$, $\delta_E$, $\delta_C$:
		dispersion-relation RMSE (in degrees) for helix, strand, and coil residues respectively;
		RMSD: backbone RMSD (\AA) of the structure reconstructed from $V_{\text{re}}$ only (setting
		$V_{\text{im}}=0$). A dash indicates that the corresponding secondary-structure type is
		absent.}
	\label{tab:decomp}
		\begin{tabular}{llcccccccr}
			\toprule
			PDB  & Class & $N$ & $\epsilon$ &
			$\langle|V_{\text{im}}|/|V_{\text{re}}|\rangle$ &
			$\rho_{\text{geom}}$ &
			$\delta_H$ & $\delta_E$ & $\delta_C$ & RMSD \\
			\midrule
			1PA7 & all-$\alpha$ & 130 & $4.9\times10^{-15}$ & 0.324 & 0.966 & 20.4 & ---   & 39.4 & 51.2 \\
			2RH3 & all-$\alpha$ & 130 & $3.1\times10^{-15}$ & 0.332 & 0.963 & 27.1 & 36.6 & 35.6 & 32.5 \\
			5SV5 & all-$\beta$  & 134 & $1.4\times10^{-14}$ & 0.290 & 0.913 & ---   & 41.7 & 40.0 & 28.9 \\
			1AYO & all-$\beta$  & 130 & $1.8\times10^{-14}$ & 0.247 & 0.971 & 20.1 & 37.1 & 40.3 & 56.6 \\
			2B1L & $\alpha/\beta$ & 129 & $5.1\times10^{-15}$ & 0.300 & 0.945 & 17.8 & 38.2 & 42.5 & 55.9 \\
			1JBE & $\alpha/\beta$ & 132 & $1.8\times10^{-15}$ & 0.305 & 0.956 & 21.3 & 45.8 & 28.8 & 36.6 \\
			3CIP & $\alpha$+$\beta$ & 130 & $9.3\times10^{-15}$ & 0.273 & 0.949 & 22.1 & 31.1 & 36.3 & 55.4 \\
			2R4I & $\alpha$+$\beta$ & 130 & $3.6\times10^{-15}$ & 0.264 & 0.969 & 24.6 & 39.4 & 32.7 & 48.2 \\
			\bottomrule
		\end{tabular}
\end{table*}

In all eight cases the error $\epsilon$ is of order $10^{-15}$ to $10^{-14}$, consistent with IEEE 754 double-precision arithmetic. This confirms that Eqs.~(\ref{eq:Vre})--(\ref{eq:Vim}) reproduce $V_{\text{eff}}$ to machine precision and contain no hidden approximation. The remaining columns of Table~\ref{tab:decomp} preview
the structural barriers analyzed in the following subsections and in Secs.~\ref{sec:piecewise}--\ref{sec:scf}; we defer their full statistical treatment to the analysis of the 856-protein dataset.

\subsection{Barrier I: imaginary potential and torsion-sign ambiguity}\label{sec:barrier1}

The imaginary part $V_{\text{im}}$ [Eq.~(\ref{eq:Vim})] depends on $\tau$ exclusively
through the sine function. Because $\sin(\tau)$ is odd, $V_{\text{im}}$ changes sign under
$\tau\to-\tau$ while $V_{\text{re}}$ [Eq.~(\ref{eq:Vre})], which depends on $\cos(\tau)$,
remains invariant. We now formalize the resulting degeneracy as a symmetry statement.
This result holds \emph{within the class of purely local,
real-potential models}, where the modeling choice is to retain only
$V_{\text{re}}=\operatorname{Re}V_{\text{eff}}$. The decomposition
[Eqs.~(\ref{eq:Vre})--(\ref{eq:Vim})] is itself an exact kinematic identity that loses
no information; the degeneracy below is a consequence of \emph{discarding}
$V_{\text{im}}$, and is therefore a property of the model class, not of the Hasimoto map.

\begin{definition}[Site-wise torsion-parity group]
\label{def:z2N}
Let $N$ be the chain length, i.e.\ the number of torsion angles $\tau[1],\dots,\tau[N]$.
Let $G=(\mathbb{Z}_2)^{N}$, whose elements are subsets $S\subseteq\{1,\dots,N\}$ with
group product the symmetric difference, generated by the elementary flips $g_k=\{k\}$.
$G$ acts on the torsion configuration by site-wise inversion,
$(g_S\cdot\tau)[k]=-\tau[k]$ if $k\in S$ and $\tau[k]$ otherwise.
\end{definition}

\begin{lemma}[Site-wise torsion-parity symmetry and chiral degeneracy]
\label{lem:z2N}
Fix the curvatures $\kappa[n]>0$ and the real bond parameters $\beta^{\pm}_n$. For every
$g_S\in G=(\mathbb{Z}_2)^N$ the real potential is invariant,
\begin{equation}
	V_{\text{re}}[n](g_S\cdot\tau)=V_{\text{re}}[n](\tau)\qquad(n=1,\dots,N),
\end{equation}
i.e.\ each $V_{\text{re}}[n]$ carries the trivial representation of $G$, whereas every
term of $V_{\text{im}}$ containing $\tau[k]$ transforms by the sign character
$\chi_k(g_S)=(-1)^{[k\in S]}$ and thus changes sign under the generator $g_k$.
Consequently, for an open (terminated) chain the configurations sharing a given profile
$\{V_{\text{re}}[n]\}_{n=1}^{N}$ and related by site-wise torsion-sign flips form a single
$G$-orbit, of size
\begin{equation}
	\#(G\cdot\tau)=2^{N'},\qquad N'=\#\{k:\tau[k]\notin\{0,\pi\}\}.
	\label{eq:degcount}
\end{equation}
For a \emph{chirally generic} backbone ($\tau[k]\notin\{0,\pi\}$ for all $k$) this is
exactly $2^{N}$.
\end{lemma}

\begin{proof}
Since $G=\langle g_1,\dots,g_N\rangle$ with commuting generators, it suffices to treat a
single generator $g_k$. By Eq.~(\ref{eq:Vre}) the angle $\tau[k]$ appears in
$V_{\text{re}}$ only in $V_{\text{re}}[k]$ (through the term
$\beta^{-}_{k}r^{-}[k]\cos\tau[k]$) and in $V_{\text{re}}[k-1]$ (through the term
$\beta^{+}_{k-1}r^{+}[k-1]\cos\tau[(k-1){+}1]=\beta^{+}_{k-1}r^{+}[k-1]\cos\tau[k]$); no
other entry depends on $\tau[k]$, and $\kappa$, $r^{\pm}$, $\beta^{\pm}$ are independent of
the torsion sign. As $g_k$ sends $\tau[k]\mapsto-\tau[k]$ and $\cos$ is even, both entries
are unchanged, proving invariance. In Eq.~(\ref{eq:Vim}) the same two slots carry
$\tau[k]$ through the odd function $\sin$, so those terms change sign; this is the sign
character $\chi_k$. The generators act on disjoint coordinates and commute, so the
statements extend to an arbitrary $g_S$ by composition. The orbit count
(\ref{eq:degcount}) then follows from the orbit--stabilizer theorem: a flip at a site with
$\tau[k]\in\{0,\pi\}$ is trivial (there $\sin\tau[k]=0$ and $\tau[k]$ lies on a symmetry
axis of $\cos$), so $|\mathrm{Stab}_G(\tau)|=2^{N-N'}$ and
$\#(G\cdot\tau)=|G|/|\mathrm{Stab}_G(\tau)|=2^{N'}$.
\end{proof}

The sign of $\tau[k]$ is the local chirality. Lemma~\ref{lem:z2N} shows that retaining
only the real potential projects every observable onto the $G$-invariant
(trivial-representation) subspace; the chiral information resides entirely in the sign
sectors carried by $V_{\text{im}}$ and is annihilated by this projection. Equivalently,
discarding $V_{\text{im}}$ is exactly the projection onto the $G$-invariants. A
right-handed $\alpha$-helix and its per-residue left-handed enantiomer therefore produce
identical $V_{\text{re}}$ profiles, and reconstruction from $V_{\text{re}}$ alone confronts
the $2^{N}$-fold ambiguity of Eq.~(\ref{eq:degcount}) for a generic backbone. This is
the rigorous content of the statement that a purely real potential cannot break the chiral
symmetry; restriction to $V_{\text{re}}$ is identical to projection onto the $G$-invariants, which are insensitive to the sign sectors.

This site-wise parity $\tau[k]\to-\tau[k]$ is a distinct object from
the discrete $\mathbb{Z}_2$ gauge transformation
$\kappa[k]\to\kappa[k]\cos\Delta_k,\ \tau[k]\to\tau[k]+\Delta_{k-1}-\Delta_k$ (with
$\Delta_k\in\{0,\pi\}$) introduced by Chernodub \textit{et al.}~\cite{chernodub2010topological}
to maintain the convention $\kappa\ge0$ across conjugacy points. The latter leaves the
Cartesian coordinates intact and is therefore a genuine gauge redundancy, whereas the
parity considered here changes the backbone chirality and hence the three-dimensional
structure, a change rendered invisible only by discarding $V_{\text{im}}$.

The degeneracy can be lifted only by adding to the energy a term that is \emph{odd} under
each $g_k$, i.e.\ that has a nonzero component in the sign sectors. The minimal local
choice is a sum of single-site odd functions,
$H_{\text{CS}}[\tau]=\lambda\sum_{k=1}^{N} f(\tau[k])$ with $f(-\theta)=-f(\theta)$, such as
the discrete total-torsion (Chern--Simons-type) coupling $\lambda\sum_k\tau[k]$. Under
$g_S$ one finds
$H_{\text{CS}}[g_S\cdot\tau]=H_{\text{CS}}[\tau]-2\lambda\sum_{k\in S}f(\tau[k])$, so for
$\lambda\neq0$ no nontrivial $g_S$ remains a symmetry: the group is explicitly broken and
minimization selects a unique handedness at each site, with $\operatorname{sign}(\lambda)$
fixing the global chirality. This is consistent with the observation of Danielsson
\textit{et al.}~\cite{danielsson2010gauge} that the torsion sign can be fixed by a
Chern--Simons term; we emphasize, however, that
such a term \emph{injects} the missing chirality information rather than conferring
predictive power, and the degeneracy is a property of the local, real-potential model
class, not of the Hasimoto map itself.

\vspace{0.5em}
\noindent\textbf{Remark (discrete probability current).}
The role of $V_{\text{im}}$ can be stated as a conservation law. Multiplying
Eq.~(\ref{eq:dnls_def}) by $\bar\psi[n]$ and taking the imaginary part gives, for
arbitrary site-dependent real $\beta^{\pm}_n$ and complex $V_{\text{eff}}$, the exact
discrete continuity identity
\begin{equation}
  \begin{aligned}
  &J[n]-J[n-1]=V_{\text{im}}[n]\,|\psi[n]|^{2},\\
  &J[n]=\beta^{+}_{n}\,\mathrm{Im}\!\bigl(\bar\psi[n]\psi[n+1]\bigr)
       =\beta^{+}_{n}\,\kappa[n]\kappa[n+1]\sin\tau[n+1],
  \end{aligned}
  \label{eq:current}
\end{equation}
where $J[n]$ is the discrete probability current (the $(\psi,\bar\psi)$ Casoratian of the
recurrence). Thus $V_{\text{im}}$ is the \emph{unique source} of $J$: $V_{\text{im}}=0$ is
precisely the condition under which $J$ would be conserved along the chain. Real backbones
have generically nonzero $V_{\text{im}}$ (median $\langle|V_{\text{im}}|/|V_{\text{re}}|\rangle
\approx 0.30$ over the $856$ chains; nonzero at $99.8\%$ of sites), so the geometric flux
$J$ is genuinely \emph{not} conserved, changing sign along every chain, and we make no
claim that real proteins obey $\kappa\kappa'\sin\tau=\text{const}$. The consequence is that
restricting to a real effective potential is not a constraint imposed on the backbone but a
\emph{model class} that is structurally blind to this flux: it discards the single
gauge-covariant, parity-odd field $V_{\text{im}}$ that sources the only nontrivial conserved
bilinear current of the discrete map. The torsion-sign degeneracy of Lemma~\ref{lem:z2N} is
the kinematic statement of this fact, and Eq.~(\ref{eq:current}) is its dynamical counterpart. The discarded
imaginary part is at once the carrier of chirality and the source of the probability current.

We quantify the information content of $V_{\text{im}}$ across the full dataset of $856$
non-redundant proteins. Figure~\ref{fig:vim_info}(a) shows the distribution of the ratio
$\langle|V_{\text{im}}|/|V_{\text{re}}|\rangle$ per protein. The mean ratio is 0.31,
indicating that the imaginary component has a magnitude roughly one-third that of the real
part. The distribution is largely independent of SCOP class, with all-$\alpha$ proteins
showing a slightly broader tail toward higher values.

To demonstrate the practical consequence of discarding $V_{\text{im}}$, we reconstruct the
backbone of each protein using $V_{\text{re}}$ only (setting $V_{\text{im}}=0$ and choosing
$\text{sign}(\tau)$ uniformly positive). Figure~\ref{fig:vim_info}(b) plots the resulting
C$\alpha$ RMSD against chain length. The RMSD grows approximately linearly, reaching
40--120~\AA{} for chains of 200--300 residues. This linear growth reflects the cumulative
nature of torsion-sign errors: each incorrectly assigned sign produces a local angular
deviation that propagates along the chain. In information-theoretic terms, each residue
contributes one bit of unresolved chiral information, for a total of $N$ bits per chain.
The $V_{\text{re}}$-only reconstruction therefore yields structures that deviate substantially from the native fold for all but the shortest peptides.

\begin{figure}[t]
	\centering
	\includegraphics[width=1\columnwidth]{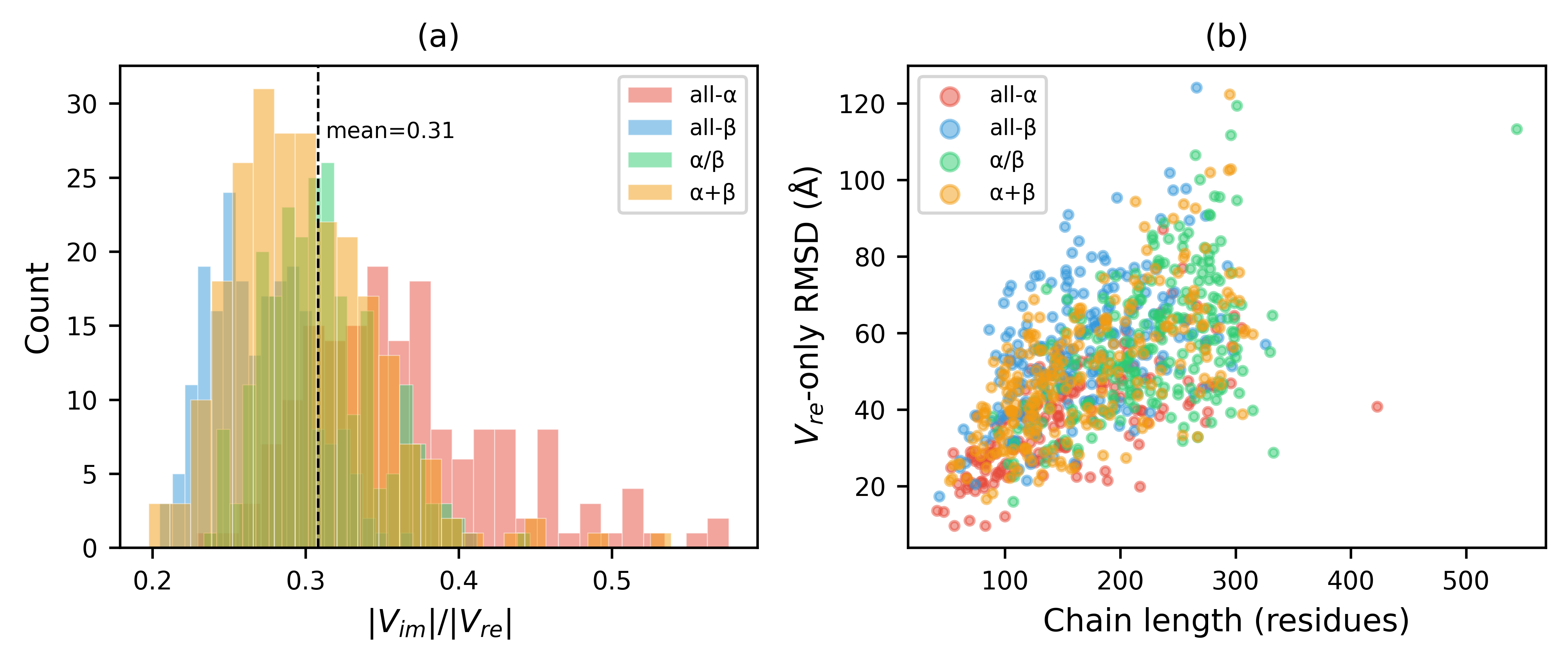}
	\caption{Information cost of discarding the imaginary potential.
		(a)~Distribution of the ratio
		$\langle|V_{\text{im}}|/|V_{\text{re}}|\rangle$ across 856
		non-redundant proteins, colored by SCOP class. The mean ratio is
		0.31, indicating that the imaginary component, which encodes the
		sign of the torsion angle through the odd symmetry of
		$\sin(\tau)$, has a magnitude roughly one-third that of the real
		part. The distribution is largely class-independent, with
		all-$\alpha$ proteins showing a slightly broader tail toward
		higher values. (b)~Backbone RMSD of structures reconstructed
		using $V_{\text{re}}$ only (setting $V_{\text{im}}=0$) versus
		chain length. RMSD grows roughly linearly with chain length,
		reaching 40--120\,\AA{} for chains of 200--300 residues. This
		reflects the cumulative effect of torsion-sign errors: each
		residue contributes 1 bit of unresolved chiral ambiguity, and
		the resulting $2^{N}$ degeneracy renders $V_{\text{re}}$-only
		reconstruction structurally unreliable beyond short peptides.}
	\label{fig:vim_info}
\end{figure}

\subsection{Barrier II: geometric dominance of $V_{\text{re}}$}\label{sec:barrier2}

The decomposition [Eq.~(\ref{eq:Vre})] shows that $V_{\text{re}}$ depends on the bond
parameters $\beta^{\pm}_{n}$ only as multiplicative prefactors of the geometric terms
$r^{\pm}\cos\tau$. A natural question is whether the sequence dependence encoded in
$\beta(s)$ contributes significantly to the spatial profile of $V_{\text{re}}$, or whether
the geometric factors dominate.

We test this by computing, for each of the 856 proteins, two versions of $V_{\text{re}}$:
one with the physical, amino-acid-dependent bond parameters $\beta^{\pm}_{n}(s)$, and one
with uniform $\beta^{+}=\beta^{-}=1$. The Spearman rank correlation $\rho_{\text{geom}}$
between the two profiles measures the fraction of the $V_{\text{re}}$ pattern that is
determined by geometry alone.

\begin{corollary}[Geometric dominance]
Across 856 non-redundant proteins, the Spearman
correlation between $V_{\text{re}}(\beta(s))$ and $V_{\text{re}}(\beta=1)$ has mean
$\bar{\rho}_{\text{geom}}=0.951$ and minimum 0.88 (Fig.~\ref{fig:geom_dom}). The residual
variance attributable to sequence-dependent $\beta(s)$ is less than $1-\bar{\rho}^{2}\approx
0.05$, i.e., below 5\% on average.
\end{corollary}

This finding indicates that sequence information is decoupled from the effective potential. Although the bond parameters $\beta^{\pm}_{n}$ depend on amino-acid identity, their contribution to the variance of $V_{\text{re}}$ is negligible. The potential is dominated by the geometric terms $r^{\pm}\cos\tau$, the fluctuations of which exceed those induced by sequence-dependent stiffness by an order of magnitude (Fig.~\ref{fig:profiles}). All
four SCOP classes overlap in the $\rho_{\text{geom}}$ distribution
(Fig.~\ref{fig:geom_dom}), confirming that this pattern is universal across fold
types. One caveat applies. The geometric terms $r^{\pm}$ and $\tau$ are
themselves determined by the amino-acid sequence through the folding process. The
Spearman test does not show that $V_{\text{re}}$ is independent of sequence in an
absolute sense; rather, it establishes that the only \emph{explicit} pathway from
sequence to $V_{\text{re}}$ within the decomposition
[Eqs.~(\ref{eq:Vre})--(\ref{eq:Vim})], namely the bond stiffnesses $\beta^{\pm}_{n}$,
is too narrow to carry significant information. The sequence dependence of $V_{\text{re}}$
is almost entirely \emph{implicit}, mediated by the three-dimensional structure
$(\kappa,\tau)$ that the folding process produces. 

\begin{figure}[t]
	\centering
	\includegraphics[width=0.8\columnwidth]{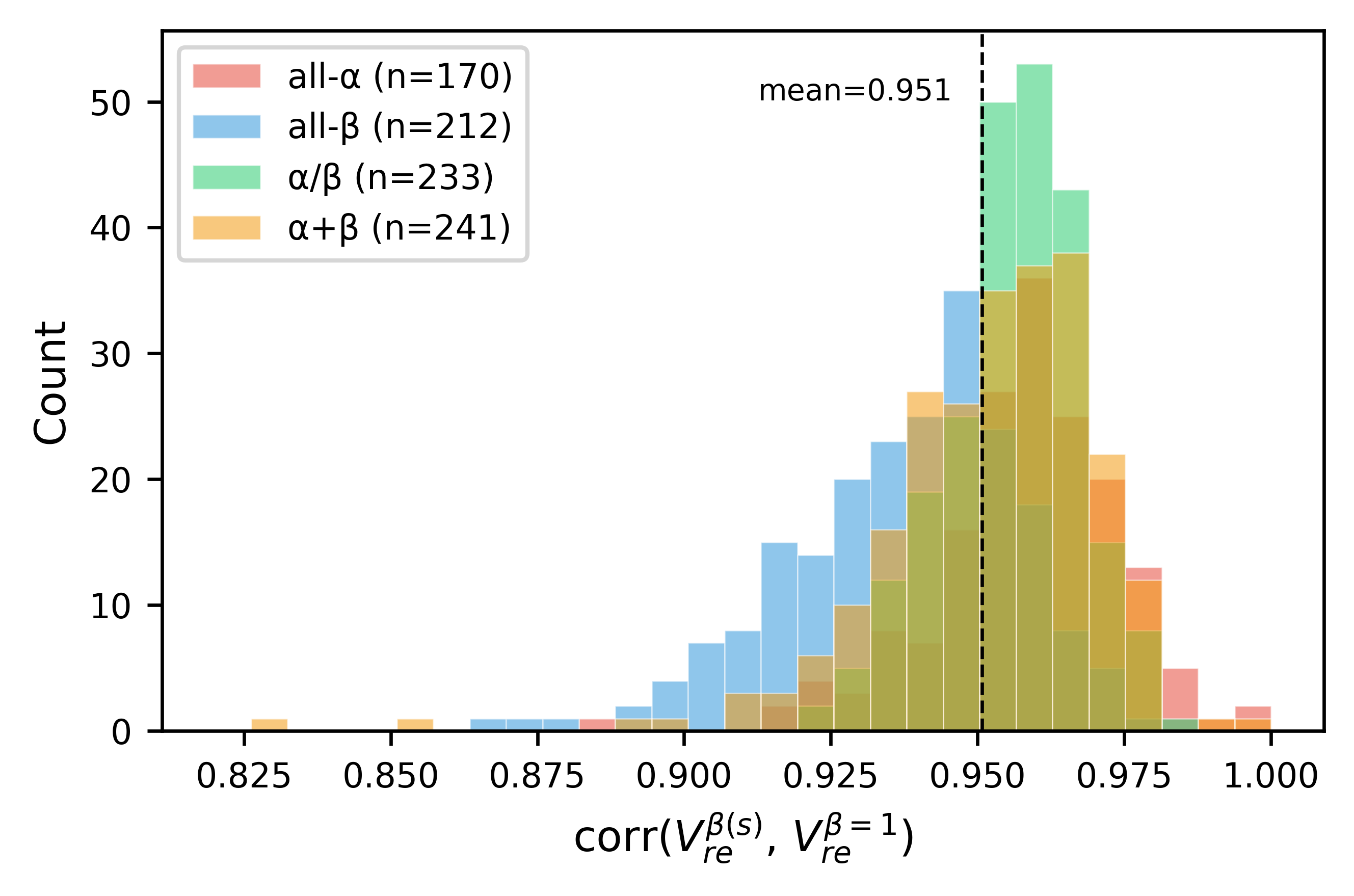}
	\caption{Geometric dominance of the real effective potential.
		Distribution of Spearman rank correlation between $V_{\text{re}}$
		computed with the physical, sequence-dependent bond parameters
		$\beta(s)$ and $V_{\text{re}}$ computed with uniform $\beta=1$,
		evaluated over 856 non-redundant proteins. Colors denote SCOP
		structural classes: all-$\alpha$ (170), all-$\beta$ (212),
		$\alpha/\beta$ (233), and $\alpha$+$\beta$ (241). The distribution
		is sharply peaked near unity (mean $= 0.951$), indicating that
		the explicit sequence dependence carried by $\beta(s)$ accounts
		for less than 5\% of the variance of $V_{\text{re}}$ on average.
		The dominant contribution comes from the geometric terms
		$r^{\pm}\cos\tau$, which depend on the backbone structure
		$(\kappa,\tau)$ rather than directly on amino-acid identity.
		All four SCOP classes overlap, confirming that this pattern is
		universal across protein folds.
		}
	\label{fig:geom_dom}
\end{figure}

To test this interpretation directly, we compared $V_{\text{re}}$ profiles between
protein pairs that share the same SCOP superfamily but have low sequence identity. From
the 856 non-redundant proteins (culled by the PISCES~\cite{wang2003pisces} server with resolution $\leq
2.0$\,\AA, $R$-factor $\leq 0.2$, chain length 40--300 residues, sequence identity
$\leq 25\%$, X-ray entries only, excluding chains with breaks or disorder), we
identified 1\,729 pairs belonging to the same SCOP superfamily. As a control, we
constructed 4\,800 pairs drawn from different SCOP folds with chain-length differences
$\leq 20$ residues. For each pair, TM-align was used to obtain a residue-level
structural alignment; the Pearson correlation $\rho_{V}$ was then computed between the
$V_{\text{re}}$ values at structurally aligned positions.

Figure~\ref{fig:fold_track} summarizes the results. Same-superfamily pairs exhibit a
mean $\rho_{V} = 0.290 \pm 0.266$, significantly higher than the different-fold
background of $\rho_{V} = 0.099 \pm 0.230$ (Mann-Whitney $U$ test,
$p < 10^{-134}$). The scatter plot [Fig.~\ref{fig:fold_track}(a)] shows that
$\rho_{V}$ increases with TM-score: among the 1\,729 same-superfamily pairs, 79\%
have TM-score $> 0.5$ and cluster in the upper-right quadrant, while different-fold
pairs remain near zero. The remaining 21\% of same-superfamily pairs fall below
TM-score $= 0.5$; these are distant homologs for which the superfamily-level
classification permits substantial structural divergence. Even within this subgroup,
$\rho_{V}$ correlates positively with TM-score, indicating that the relationship
between structural similarity and $V_{\text{re}}$ similarity is continuous rather than
threshold-dependent. The mean sequence identity within the same-superfamily group is
only 13.9\% (range 1.4--27.9\%), confirming that the elevated $\rho_{V}$ is driven by
structural similarity rather than sequence similarity. This provides direct evidence
that $V_{\text{re}}$ tracks the three-dimensional fold: proteins with unrelated
sequences but similar structures produce similar $V_{\text{re}}$ profiles, whereas
proteins with comparable chain lengths but different folds do not.

\begin{figure}[t]
	\centering
	\includegraphics[width=1\columnwidth]{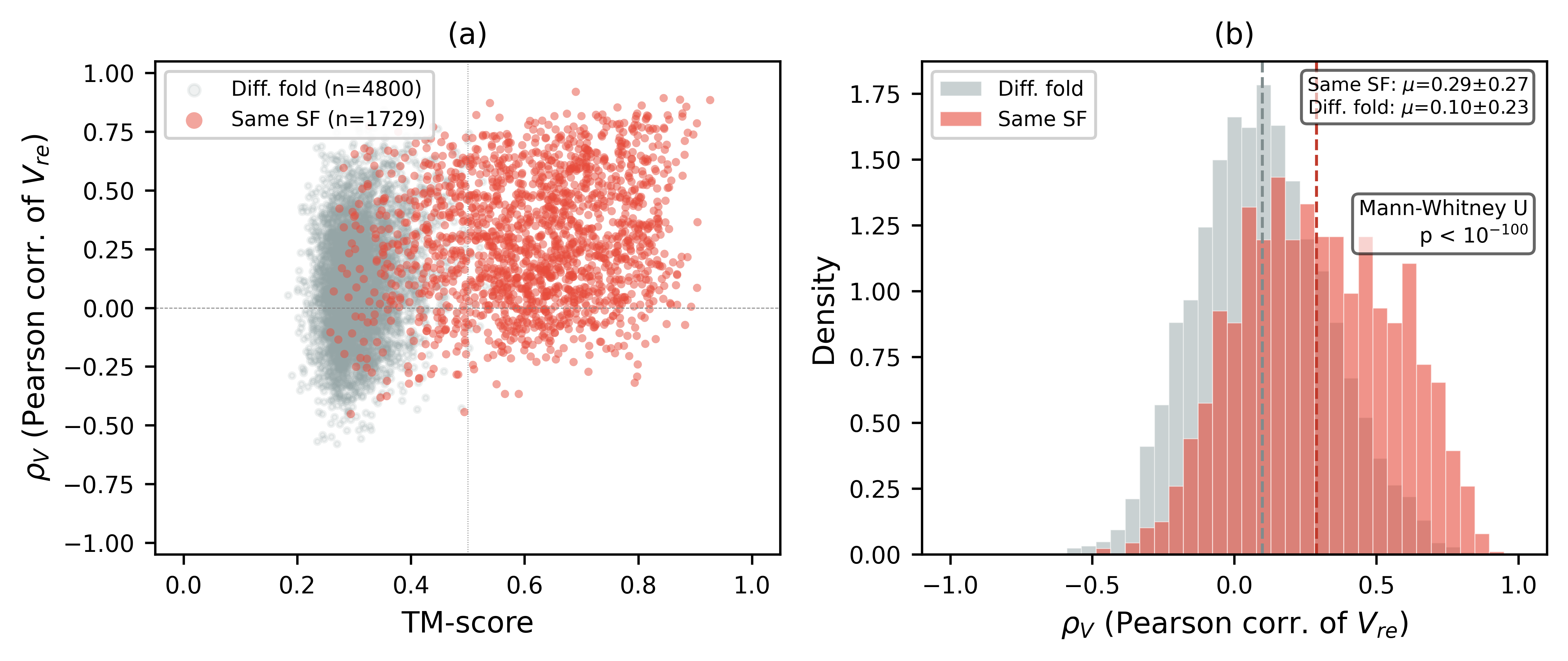}
	\caption{$V_{\text{re}}$ tracks fold rather than sequence.
		(a)~Pearson correlation $\rho_{V}$ of structurally aligned
		$V_{\text{re}}$ profiles versus TM-score for 1\,729
		same-superfamily pairs (red) and 4\,800 different-fold pairs
		(gray). Among the same-superfamily pairs, 79\% have
		TM-score $> 0.5$ and cluster in the upper-right quadrant;
		the remaining 21\% are distant homologs with greater
		structural divergence, yet $\rho_{V}$ still correlates
		positively with TM-score within this subgroup.
		(b)~Distribution of $\rho_{V}$ for the two groups.
		Same-superfamily: $\mu = 0.29 \pm 0.27$; different-fold:
		$\mu = 0.10 \pm 0.23$ (Mann-Whitney $U$,
		$p < 10^{-134}$). The mean sequence identity within the
		same-superfamily group is 13.9\%, confirming that the elevated
		correlation is driven by structural similarity, not sequence
		similarity.}
	\label{fig:fold_track}
\end{figure}

The exact decomposition [Eq.~(\ref{eq:Vre})] establishes that $V_{\text{re}}$ is driven by two distinct inputs: the geometric terms $r^{\pm}\cos\tau$, which depend on the three-dimensional structure, and the bond stiffnesses $\beta^{\pm}_{n}$, which depend directly on the amino-acid sequence. However, the Spearman test demonstrates that the sequence-dependent terms contribute less than 5\% of the variance, and the superfamily analysis (Fig.~\ref{fig:fold_track}) confirms that $V_{\text{re}}$ covaries with the fold rather than the sequence. Thermodynamically, the native state is determined by the chemical identity of the amino-acid sequence. In the DNLS formulation, this causal relationship is effectively inverted: the spatial geometry dominates the effective potential, suppressing the explicit contribution of amino-acid chemistry. This structural circularity limits forward prediction: $V_{\text{re}}$ is overwhelmingly determined by the target backbone geometry $(\kappa,\tau)$ one seeks to predict, while the direct pathway from sequence to $V_{\text{re}}$ through $\beta^{\pm}_{n}$ is mathematically too weak to carry the necessary folding information.

The circularity exposed by Barrier~II can be contrasted with
physical models that have achieved predictive success precisely by
incorporating non-local, sequence-dependent information through
channels absent from the Hasimoto decomposition. Coarse-grained
force fields such as AWSEM~\cite{davtyan2012awsem} supplement local
backbone terms with explicit contact potentials that couple residue
pairs separated by tens to hundreds of positions along the sequence,
and with associative-memory terms that bias the energy landscape
toward known structural motifs. The tube model of Banavar and
Maritan~\cite{banavar2007physics, banavar2002geometry} derives
secondary-structure selection from a three-body excluded-volume
interaction that is non-local by construction. Statistical approaches
based on direct coupling analysis~\cite{morcos2011direct,
	cocco2013principal} extract residue--residue contact maps from
evolutionary covariance using the maximum-entropy principle, providing
non-local structural constraints derived entirely from sequence data.
In each case, the predictive power arises from terms that couple residues at positions $|m-n| \gg 1$, which constitutes precisely the type of information that the DNLS effective potential, constructed from nearest-neighbor ratios $r^{\pm}[n]$ and local torsion angles $\tau[n]$, cannot encode. The geometric dominance of $V_{\text{re}}$
is therefore not merely a quantitative observation but reflects a
structural mismatch between the local, kinematic content of the
nearest-neighbor DNLS reduction and the non-local, thermodynamic content required for
folding.

\subsection{Classification of folding forces by parity and locality}
\label{sec:parity_locality}

Barriers~I and~II locate the limitation of the real-potential DNLS in two algebraic
properties of the decomposition: the chiral information resides in the parity-odd part
$V_{\text{im}}$, and the only explicit sequence channel in the parity-even part
$V_{\text{re}}$ is the near-uniform nearest-neighbor stiffness $\beta^{\pm}_n$. These
two axes, parity under $\tau\to-\tau$ and locality along the chain, organize the
physical interactions that drive folding into the four sectors of
Table~\ref{tab:parity_locality}. The placement of each interaction follows from an
elementary symmetry argument: a term contributes to $V_{\text{im}}$ (is parity-odd) only
if it changes sign under $\tau\to-\tau$, which requires explicit chirality; and it is
non-local only if it couples residues at sequence separation $|m-n|>1$. Site-wise
inversion $\tau\to-\tau$ is a reflection isometry that preserves all interatomic
distances $|\mathbf{r}_i-\mathbf{r}_j|$, so every distance-dependent interaction
(Coulomb, dispersion, disulfide, $\pi$-stacking, solvent burial) is parity-even and
cannot, by itself, break the chiral degeneracy of Barrier~I; only the stereochemistry of
the L-amino-acid side chain (the C$_\beta$ branch and the C$_\alpha$ improper dihedral)
supplies a parity-odd, $h\sin\tau$-type term.

\begin{table*}[t]
	\centering
	\caption{Classification of the principal interactions governing protein structure by
		parity under site-wise torsion inversion $\tau\to-\tau$ (columns) and by locality
		along the backbone (rows). The current local, real-potential DNLS occupies only the
		even$\times$local cell. Chirality requires the parity-odd column ($V_{\text{im}}$, a
		local addition); every driver of tertiary structure lies in the non-local row, which
		a nearest-neighbor scalar potential cannot represent.}
	\label{tab:parity_locality}
	\renewcommand{\arraystretch}{1.3}
	\setlength{\tabcolsep}{8pt}
	\begin{tabular}{>{\raggedright\arraybackslash}p{2.2cm} >{\raggedright\arraybackslash}p{7.2cm} >{\raggedright\arraybackslash}p{5.2cm}}
		\toprule
		& \textbf{Even} ($V_{\text{re}}$) & \textbf{Odd} ($V_{\text{im}}$, chiral) \\
		\midrule
		\textbf{Local} & backbone bending $\kappa$; van der Waals repulsion; excluded volume; $\omega$ planarity \emph{(current DNLS)} & C$_\beta$ stereochemistry; C$_\alpha$ improper ($h\sin\tau$) \\
		\addlinespace
		\textbf{Non-local} & solvent burial; Coulomb / salt bridges; disulfides; London dispersion; $\pi$/cation/metal; hydrogen bonds & chiral non-local packing \\
		\bottomrule
	\end{tabular}
\end{table*}

Three consequences follow. First, the current model populates only the even$\times$local
cell; the chiral sector is a \emph{local} addition (restoring $V_{\text{im}}$, or
equivalently a Chern--Simons-type $\sum_k\tau[k]$ term as in Barrier~I) and is therefore
within reach of the framework. Second, and most important, every interaction that selects
\emph{tertiary} structure (hydrophobic burial, electrostatics, disulfide cross-links,
dispersion, and the $i\to i+4$ hydrogen bond that stabilizes the helix) occupies the
non-local row, which a potential built from nearest-neighbor ratios $r^{\pm}[n]$ cannot
encode without promotion to a range-$L$ operator. This is the structural content of the
requirement that a predictive model incorporate interactions with the
external environment: in the present classification those interactions are precisely the
non-local, even-parity sector (solvent and electrostatics), absent from the local DNLS by
construction rather than by parameter choice. Third, among these non-local forces the
bare Coulomb interaction is distinguished: its $1/r$ lattice dispersion
$\sum_{m\ge1}\cos(qm)/m=-\ln|2\sin(q/2)|$ is the unique \emph{non-analytic} case among the
forces examined, since a $1/r^{p}$ tail yields a non-analytic long-wavelength dispersion if and
only if $p\le1$, and only Coulomb saturates this bound. Physiologically, however, Debye
screening cuts the divergence at $\lambda_D\approx10$\,\AA, comparable to the chiral
correlation length $\xi\approx4$ residues, so screened electrostatics sits \emph{at the
boundary} between genuinely long-range and effectively short-range rather than dominating
the global modes.

A uniform integrable segment survives only while the couplings are translation-invariant.
Promoting any non-local sector to a sequence-dependent, range-$L$ operator, as a
predictive model must, renders the effective couplings site-dependent and destroys the
uniform, nearest-neighbor structure on which the soliton reductions of the DNLS rest
(Sec.~\ref{sec:piecewise}): the integrable reduction and the
non-local content required for folding are mutually exclusive within this scalar
framework. The local DNLS is thus a faithful kinematic descriptor of backbone geometry,
not an \textit{ab initio} folding engine; the forces that fold the chain live in sectors
it does not reach, and reaching them forfeits the integrable structure that makes it
tractable.

The exact decomposition has, in summary, revealed two independent, static barriers to forward structure
prediction: the $2^{N}$-fold torsion-sign degeneracy encoded in $V_{\text{im}}$
(Barrier~I) and the geometric dominance of $V_{\text{re}}$ that leaves less than 5\% of
its variance attributable to amino-acid identity (Barrier~II). Both barriers follow
from the algebraic structure of the real-potential decomposition
[Eqs.~(\ref{eq:Vre})--(\ref{eq:Vim})], specifically from discarding
$V_{\text{im}}$ and from the narrowness of the $\beta^{\pm}_n$ channel, rather
than from the Hasimoto transform itself, which preserves the full
$(\kappa,\tau)$ information. Whether the DNLS can nevertheless function as a
dynamical equation that drives folding through a physically motivated effective potential
is tested in Sec.~\ref{sec:scf}.

\section{Piecewise integrability and secondary structure}
\label{sec:piecewise}

The exact decomposition of Sec.~\ref{sec:decomposition} holds for arbitrary discrete
curves. In this section we examine the special case of backbone segments where the
curvature varies slowly, so that $r^{\pm}[n]\approx 1$. Under this condition the full
decomposition simplifies to a scalar dispersion relation that connects $V_{\text{re}}$
directly to the torsion angle. We show that this dispersion relation is satisfied
almost exclusively within $\alpha$-helical segments, providing a purely geometric
criterion for secondary-structure identification.

\subsection{Uniform-segment dispersion relation}

When the curvature ratios satisfy $r^{+}[n]\approx r^{-}[n]\approx 1$ and the bond
parameters are approximately uniform ($\beta^{+}_{n}\approx\beta^{-}_{n}\approx\beta$),
the real part of the effective potential [Eq.~(\ref{eq:Vre})] reduces to
\begin{equation}
	V_{\text{re}}[n] \approx \beta\bigl[\cos\tau[n+1]+\cos\tau[n]\bigr] - 2\beta\,.
	\label{eq:Vre_uniform}
\end{equation}
If the torsion angle is also locally constant ($\tau[n+1]\approx\tau[n]\approx\tau$),
this further simplifies to the dispersion relation
\begin{equation}
	\cos\tau = 1 + \frac{V_{\text{re}}}{2\beta}\,.
	\label{eq:dispersion}
\end{equation}
Eq.~(\ref{eq:dispersion}) is the discrete analogue of the continuum NLS dispersion
relation $\omega = k^{2}$, expressed in terms of the backbone torsion angle. It provides
a one-to-one mapping between $V_{\text{re}}$ and $|\tau|$ within any segment where
$\kappa$ and $\tau$ are approximately uniform. Note that the cosine function renders
Eq.~(\ref{eq:dispersion}) insensitive to the sign of $\tau$, consistent with the
torsion-sign degeneracy identified in Sec.~\ref{sec:barrier1}.

\vspace{0.5em}
\noindent\textbf{Transfer-matrix interpretation.}
The uniform-segment reduction has a dynamical-systems reading that makes the
integrability structure explicit. Writing the constant-coefficient recurrence
$\psi[n+1]=(2+V_{\text{re}}/\beta)\psi[n]-\psi[n-1]$ for the state
$\Psi_n=(\psi[n],\psi[n-1])^{\mathsf T}$,
\begin{equation}
  \Psi_{n+1}=T\,\Psi_n,\qquad
  T=\begin{pmatrix} 2+\dfrac{V_{\text{re}}}{\beta} & -1\\[2pt] 1 & 0\end{pmatrix}
  \in SL(2,\mathbb{R}).
  \label{eq:transfer}
\end{equation}
Since $\det T=1$, the map is area-preserving; its eigenvalues form the reciprocal pair
$e^{\pm iq}$ with the Bloch dispersion $\cos q=\tfrac12\operatorname{Tr}T$, i.e.\
$\cos q=1+V_{\text{re}}/2\beta$,
identical to Eq.~(\ref{eq:dispersion}). For the helical value
$V_{\text{re}}=2\beta(\cos\tau-1)$, which is Eq.~(\ref{eq:dispersion}) rearranged, this
reduces to $\cos q=\cos\tau$, so the Bloch quasi-momentum coincides with the torsion angle
\emph{in magnitude}, $|q|=|\tau|$, within the band $-4\beta\le V_{\text{re}}\le 0$.
A uniform $\alpha$-helix is then an \emph{elliptic} fixed point of $T$
($|\operatorname{Tr}T|<2$, $q$ real), about which the segment is a marginally stable
Bloch mode; the band edges $q\in\{0,\pi\}$ ($V_{\text{re}}=0$ or $-4\beta$) are the
parabolic transitions beyond which the fixed point becomes hyperbolic.

The identification is necessarily unsigned. The transfer matrix~(\ref{eq:transfer}) is
built from the real potential $V_{\text{re}}$ alone, which enters only through $\cos\tau$
and is invariant under the site-wise parity $\tau\to-\tau$ of Lemma~\ref{lem:z2N}; its
eigenpair $\{e^{+iq},e^{-iq}\}$ likewise contains both signs, so neither the dispersion
relation nor $T$ selects $\operatorname{sign}q$. The signed phase
$\psi[n+1]/\psi[n]=e^{i\tau}$ of the full Hasimoto field is recovered only through the
imaginary potential $V_{\text{im}}$, and the residual sign of $q$ is precisely the chiral
degree of freedom of Barrier~I. The transfer-matrix analysis thus independently recovers
and reinforces Lemma~\ref{lem:z2N}: the local linear integrability of a helical segment
fixes the quasi-momentum only up to $|q|=|\tau|$, so the chiral degeneracy is woven into
the integrable structure itself.

The conditions under which Eq.~(\ref{eq:dispersion}) holds are precisely the conditions
that define an integrable segment of the DNLS: uniform curvature ($r^{\pm}\approx 1$),
uniform torsion, and uniform coupling. Deviations from these conditions break
integrability and cause the full decomposition [Eqs.~(\ref{eq:Vre})--(\ref{eq:Vim})]
to differ from the dispersion relation. The magnitude of this deviation therefore
serves as a local measure of integrability.

\subsection{Integrability error as a structural probe}

We define the integrability error at each residue as
\begin{equation}
	E[n] = \left|\cos\tau[n] - \left(1 + \frac{V_{\text{re}}[n]}{2\beta}\right)\right|,
	\label{eq:E_def}
\end{equation}
where $\beta = \langle\beta^{\pm}_{n}\rangle$ is the chain-averaged coupling parameter.
By construction, $E[n]=0$ when the backbone at residue $n$ satisfies the uniform-segment
dispersion relation exactly, and $E[n]>0$ when the local curvature or torsion varies
too rapidly for the integrable approximation to hold.

Figure~\ref{fig:profiles} shows $V_{\text{eff}}[n]$ along the backbone for eight representative proteins (two per SCOP class); we treat the effective potential as a site-dependent profile analogous to spectral signal representations of protein sequences~\cite{wang2025signal}. Within helical segments (pink background shading), $V_{\text{re}}$ forms near-constant negative plateaus (typically $-0.5$ to $-1.5$\,\AA$^{-1}$), consistent with Eq.~(\ref{eq:dispersion}) evaluated at the canonical helix torsion $\tau_{\text{helix}}\approx 1.0$\,rad. The amplitude of $V_{\text{im}}$ is reduced relative to loop and strand regions but remains finite. In $\beta$-strand and coil
regions, both $V_{\text{re}}$ and $V_{\text{im}}$ fluctuate strongly (amplitudes
reaching 3--6\,\AA$^{-1}$), reflecting rapid residue-to-residue variation of $r^{\pm}$
and $\tau$. Sharp negative spikes in $V_{\text{re}}$ mark transitions between
secondary-structure elements. These patterns are local rather than class-dependent:
helical segments in the nominally all-$\beta$ protein 1AYO display the same plateau
behavior as those in the all-$\alpha$ proteins.

\begin{figure*}[t]
	\centering
	\includegraphics[width=\textwidth]{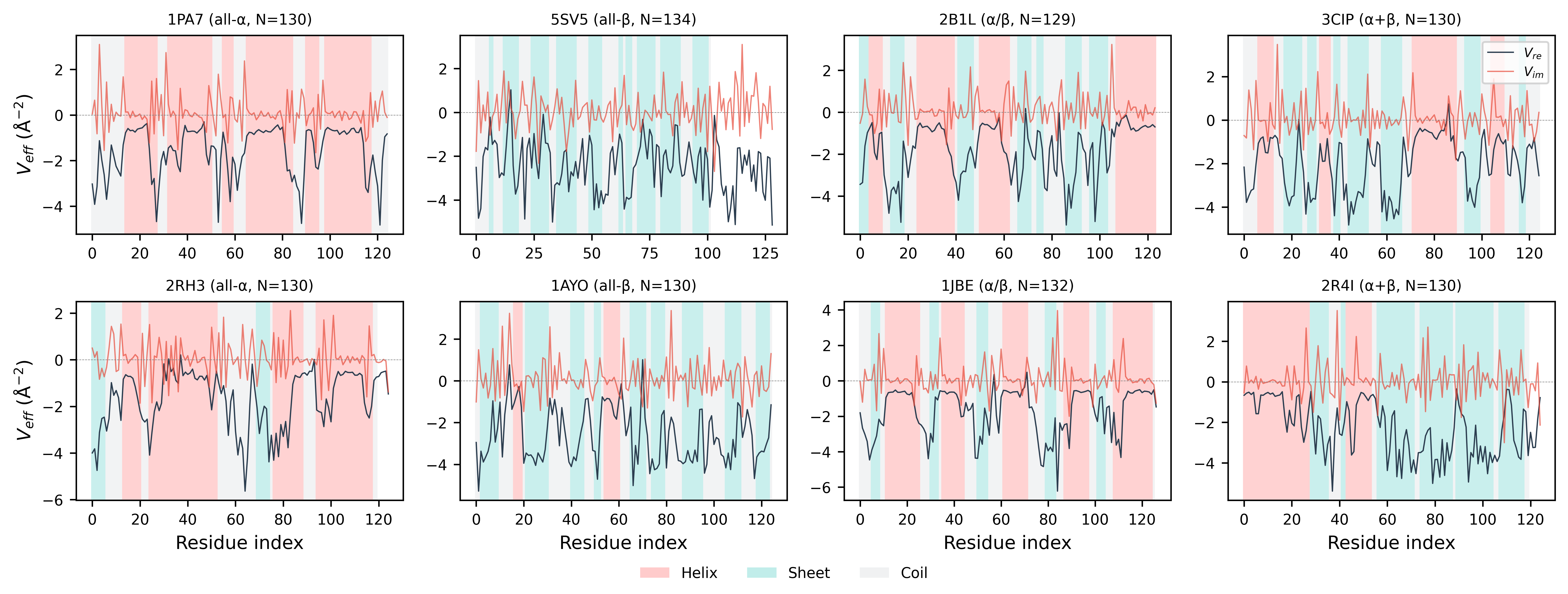}
	\caption{Effective potential $V_{\text{eff}}[n]$ along the C$_\alpha$ backbone for
		eight representative proteins (two per SCOP class; columns from left to right:
		all-$\alpha$, all-$\beta$, $\alpha/\beta$, $\alpha$+$\beta$). Black:
		$V_{\text{re}}$; red: $V_{\text{im}}$. Background shading marks DSSP
		secondary-structure assignment (pink: helix; cyan: strand; gray: coil). Within
		helical segments $V_{\text{re}}$ forms near-constant negative plateaus consistent
		with the dispersion relation Eq.~(\ref{eq:dispersion}), while strand and coil
		regions exhibit large-amplitude fluctuations in both components. Sharp negative
		spikes in $V_{\text{re}}$ mark transitions between secondary-structure elements.
		These patterns are local rather than class-dependent: helical segments in the
		all-$\beta$ protein 1AYO display the same plateau behavior as those in the
		all-$\alpha$ proteins.}
	\label{fig:profiles}
\end{figure*}

\subsection{Statistical validation on 856 proteins}

We evaluate the dispersion-relation RMSE separately for helix (H), strand (E), and
coil (C) residues across the full dataset of 856 non-redundant proteins. For each
protein, residues are grouped by their DSSP secondary-structure assignment~\cite{kabsch1983dictionary}, and the
RMSE of Eq.~(\ref{eq:dispersion}) is computed per group:
\begin{equation}
	\delta_{X} = \sqrt{\frac{1}{N_{X}}\sum_{n\in X}
		\left[\cos\tau[n] - \left(1+\frac{V_{\text{re}}[n]}{2\beta}\right)\right]^{2}}\,,
	\label{eq:rmse_def}
\end{equation}
where $X\in\{H,E,C\}$ and $N_{X}$ is the number of residues in class $X$.

Figure~\ref{fig:dispersion} shows the distribution of $\delta_{X}$ as box plots,
faceted by SCOP class. Across all four classes, helical segments exhibit systematically
lower RMSE (median ${\sim}21^{\circ}$) than strand (${\sim}37^{\circ}$) or coil
(${\sim}36^{\circ}$) segments. This separation is class-independent: even in
all-$\beta$ proteins where helices are scarce, the few helical residues still satisfy
the dispersion relation with comparable accuracy. The result reflects the local
uniformity of $(\kappa,\tau)$ within helices ($r^{\pm}\approx 1$), which is the
condition under which the exact decomposition reduces to Eq.~(\ref{eq:dispersion}).

\begin{figure*}[t]
	\centering
	\includegraphics[width=0.8\textwidth]{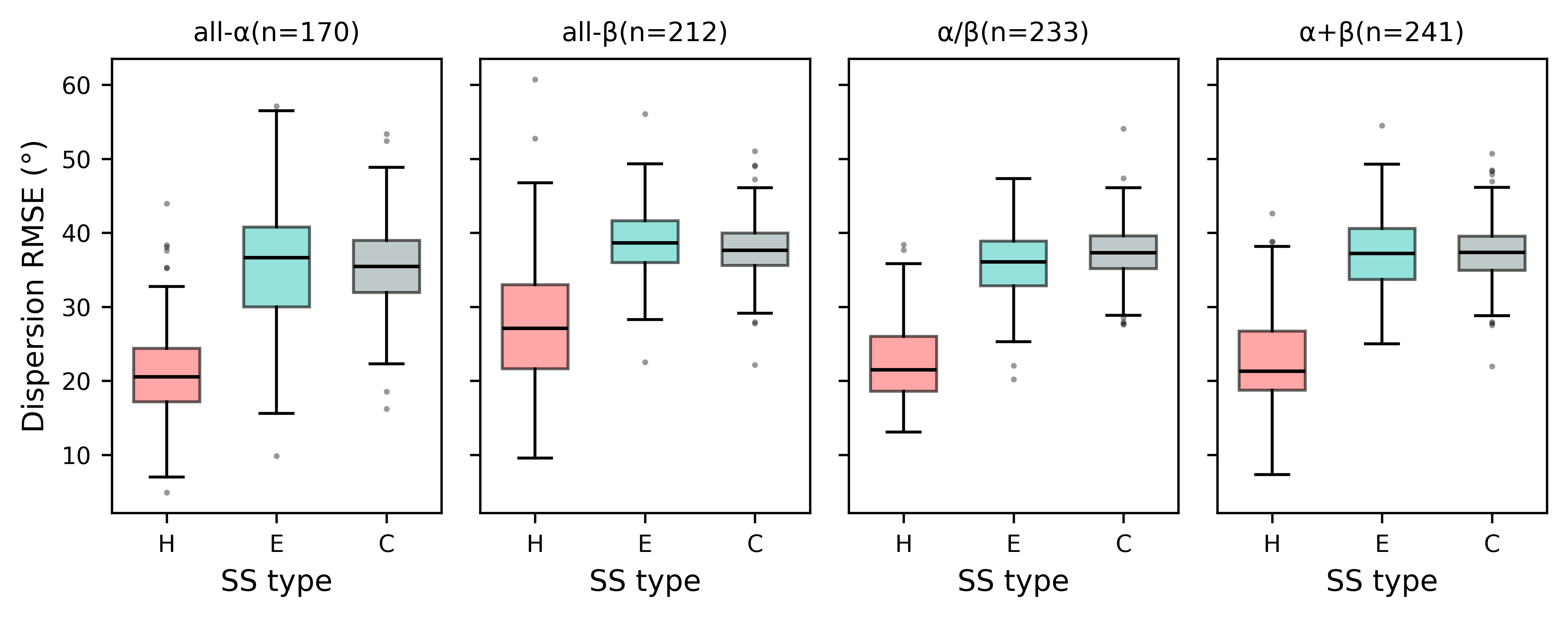}
	\caption{Dispersion-relation RMSE by secondary-structure type, faceted by SCOP
		class (856 non-redundant proteins). For each protein, residues are grouped by DSSP
		assignment into helix (H), strand (E), and coil (C), and the RMSE of the
		uniform-segment approximation $\cos\tau = 1 + V_{\text{re}}/2\beta$ is computed per
		group. Helical segments exhibit systematically lower RMSE (median ${\sim}21^{\circ}$)
		than strand (${\sim}37^{\circ}$) or coil (${\sim}36^{\circ}$) segments across all
		four SCOP classes. This separation is class-independent: even in all-$\beta$ proteins
		where helices are scarce, the few helical residues satisfy the dispersion relation
		with comparable accuracy.}
	\label{fig:dispersion}
\end{figure*}

The data in Table~\ref{tab:decomp} illustrate this pattern at the level of individual
proteins. The pure all-$\alpha$ protein 1PA7 has $\delta_{H}=20.4^{\circ}$ and no
strand residues, while the pure all-$\beta$ protein 5SV5 has $\delta_{E}=41.7^{\circ}$
and no helix residues. In mixed-class proteins, the helix and strand RMSE values
coexist within the same chain: 1AYO (all-$\beta$ by SCOP classification) contains a
small number of helical residues with $\delta_{H}=20.1^{\circ}$, comparable to the
values observed in all-$\alpha$ proteins. This confirms that the applicability of the
dispersion relation is determined by the local curvature uniformity at each residue,
not by the global fold classification.

\subsection{Helix detection by integrability error}

The systematic separation of $\delta_{H}$ from $\delta_{E}$ and $\delta_{C}$ suggests
that the integrability error $E[n]$ [Eq.~(\ref{eq:E_def})] can serve as a binary
classifier for helical residues. To quantify this, we construct a receiver operating
characteristic (ROC) curve by varying a threshold $E_{\text{th}}$ and classifying
residue $n$ as helical if $E[n] < E_{\text{th}}$ and as non-helical otherwise. The
ground truth is provided by DSSP assignment.

Figure~\ref{fig:roc} shows the ROC curves for each SCOP class and for the full dataset
(856 proteins, 143\,202 residues). The global area under the curve (AUC) is 0.720,
significantly above the random baseline of 0.5. The class dependence is weak: AUC
ranges from 0.667 (all-$\beta$) to 0.739 (all-$\alpha$), with $\alpha/\beta$ and
$\alpha$+$\beta$ at 0.713 and 0.714 respectively. Even in all-$\beta$ proteins, where
helical residues constitute only 15\% of the total (4\,741 out of 30\,980), the AUC
remains well above chance. This confirms that the integrability error operates at the
residue level, detecting local geometric regularity rather than relying on the overall
fold type.

\begin{figure}[t]
	\centering
	\includegraphics[width=1\columnwidth]{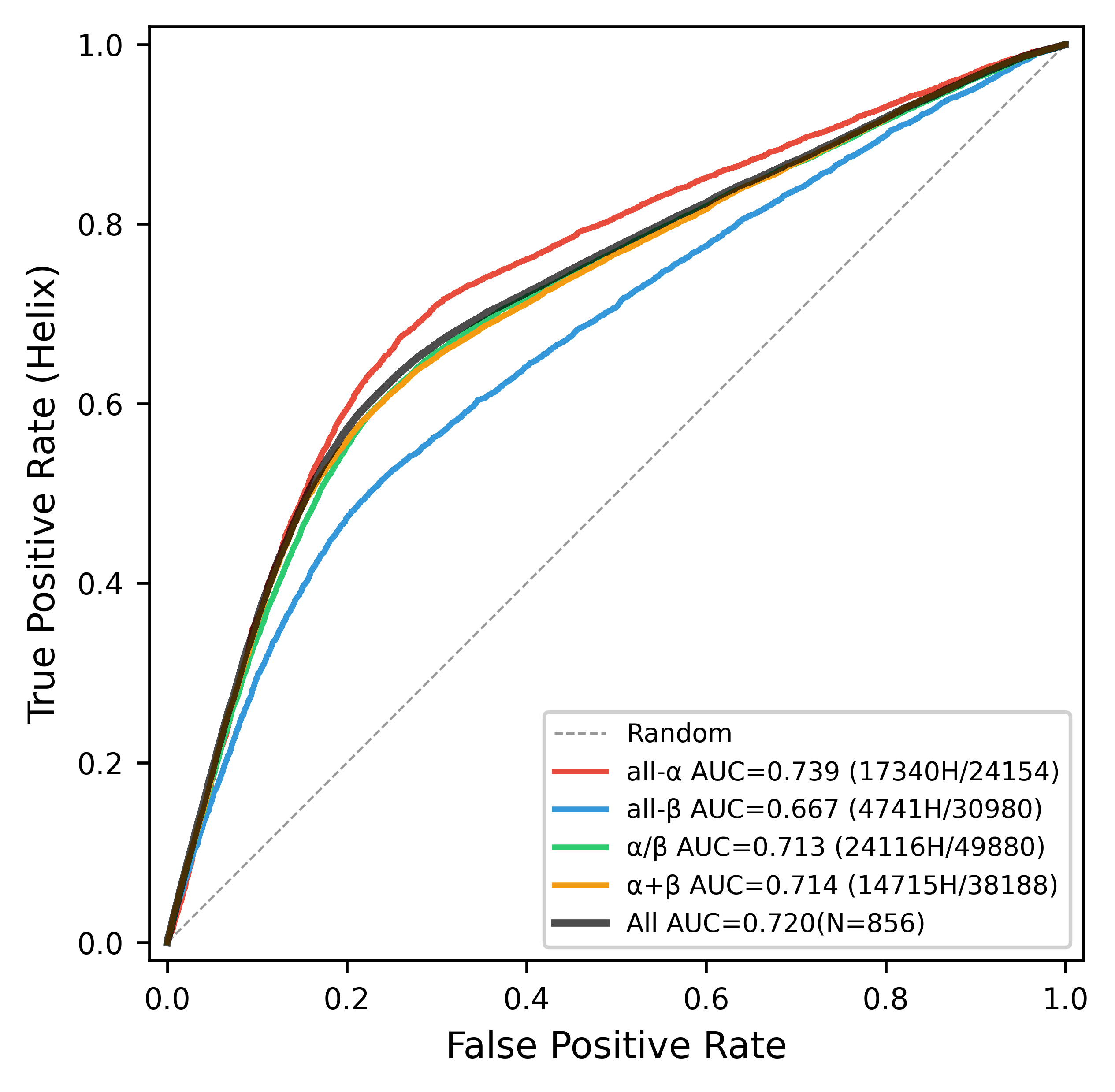}
	\caption{ROC curve for helix detection using the integrability error
		$E[n]=|\cos\tau[n]-(1+V_{\text{re}}[n]/2\beta)|$ as a binary classifier (helix
		vs.\ non-helix by DSSP). Curves are shown for each SCOP class and for the full
		dataset ($856$ proteins, 143\,202 residues). Parentheses indicate the number of
		helix residues over total residues in each class. The global AUC $= 0.720$ means that a randomly chosen helical residue has lower $E[n]$ than a randomly chosen non-helical residue 72\% of the time. The class dependence is weak
		(AUC range 0.667--0.739), confirming that the integrability--helicity correspondence
		operates at the residue level rather than depending on global fold type. The gap
		between AUC $= 0.72$ and unity quantifies the structural information carried by
		hydrogen bonds, non-local contacts, and torsion-sign degrees of freedom that lie
		outside the scalar dispersion relation.
		}
	\label{fig:roc}
\end{figure}

The condition
$E[n] < E_{\text{th}}$ tests whether the backbone at residue $n$
preserves discrete helical symmetry, that is, whether the local curvature
and torsion are sufficiently uniform that a single screw motion
$\mathbf{r}[n+1] = M\,\mathbf{r}[n]$ approximately generates the
chain segment. When $\kappa[n]$ and $\tau[n]$ are exactly constant,
the curvature ratios satisfy $r^{\pm}=1$, the exact decomposition
[Eq.~(\ref{eq:Vre})] reduces to the dispersion relation
[Eq.~(\ref{eq:dispersion})], and $E[n]\equiv 0$; the geometric origin
of this vanishing is analyzed in Sec.~\ref{sec:lancret}, where we show
that $E[n]$ is an order parameter for discrete helical symmetry breaking. In this sense, $\alpha$-helices are the backbone regions where the DNLS
most closely approximates an integrable system, and $E[n]$ measures
the degree to which this integrability is locally broken.

\subsection{The protein backbone as a piecewise integrable system}

The results of this section lead to a unified geometric picture of protein backbone
structure in the DNLS framework. The backbone can be viewed as a piecewise integrable
system: $\alpha$-helical segments are regions of approximate integrability where the
curvature and torsion are locally uniform, the dispersion relation
[Eq.~(\ref{eq:dispersion})] holds, and the DNLS admits soliton-like solutions.
$\beta$-strand and coil regions, by contrast, are regions of broken integrability where
$r^{\pm}$ deviates significantly from unity, the dispersion relation fails, and the
full decomposition [Eqs.~(\ref{eq:Vre})--(\ref{eq:Vim})] is required.

This picture provides a geometric definition of secondary structure that is independent
of hydrogen-bond criteria. An $\alpha$-helix is a contiguous segment where the DNLS
integrability error $E[n]$ remains below a threshold; a non-helical region is one where
$E[n]$ exceeds this threshold. The definition is purely kinematic: it depends only on
the C$_\alpha$ coordinates and requires no energy function or force field.

Helical segments exhibit the smallest dispersion-relation error among all
secondary-structure types (Fig.~\ref{fig:dispersion}), but they are not the sole
regions of approximate integrability. Short segments within $\beta$-strands or loops
may also satisfy $r^{\pm}\approx 1$ locally, contributing to the imperfect separation
in the ROC analysis (Fig.~\ref{fig:roc}). The gap between AUC $= 0.72$ and unity
reflects both these local exceptions and the information carried by $V_{\text{im}}$
(torsion-angle signs) and non-local interactions (hydrogen bonds, tertiary contacts)
that lie outside the scalar dispersion relation.

The dispersion relation has two limitations as a structural probe.
First, Eq.~(\ref{eq:dispersion}) determines only $|\tau|$, not $\text{sign}(\tau)$.
The torsion-sign ambiguity identified in Sec.~\ref{sec:barrier1} therefore persists even within
integrable segments. Second, the dispersion relation provides no information about the
spatial arrangement of secondary-structure elements relative to one another. It is a
local diagnostic that characterizes individual residues but cannot address the global
fold topology. These limitations reinforce the conclusion that the local, real-potential reduction is a
geometric analysis tool rather than a predictive framework.

\subsection{Geometric interpretation of the integrability error}
\label{sec:lancret}

The integrability error $E[n]$ defined in Eq.~(\ref{eq:E_def}) admits a
direct geometric interpretation as a measure of broken discrete helical
symmetry. We make this connection explicit.

A discrete curve $\{\mathbf{r}[n]\}$ possesses \emph{discrete helical
	symmetry} if there exists a rigid screw motion $M$ (a rotation about a
fixed axis composed with a translation along that axis) such that
$\mathbf{r}[n+1] = M\,\mathbf{r}[n]$ for all $n$. This condition
requires that the discrete Frenet curvature and torsion be
site-independent:
\begin{equation}
	\kappa[n] = \kappa_{0},\quad \tau[n] = \tau_{0},
	\quad \forall\; n.
	\label{eq:helix_sym}
\end{equation}
Substituting Eq.~(\ref{eq:helix_sym}) into the exact decomposition
[Eq.~(\ref{eq:Vre})] with uniform coupling
$\beta^{+}_{n}=\beta^{-}_{n}=\beta$ gives $r^{\pm}[n]=1$ and
$\tau[n+1]=\tau[n]=\tau_{0}$, so that
\begin{equation}
	V_{\text{re}} = \beta\cos\tau_{0} + \beta\cos\tau_{0} - 2\beta
	= 2\beta(\cos\tau_{0}-1),
	\label{eq:Vre_helix}
\end{equation}
which is precisely the dispersion relation [Eq.~(\ref{eq:dispersion})].
The integrability error therefore vanishes identically:
\begin{equation}
	E[n] = \left|\cos\tau_{0}
	- \left(1+\frac{2\beta(\cos\tau_{0}-1)}{2\beta}\right)\right|
	= 0.
	\label{eq:E_zero}
\end{equation}

The uniform condition Eq.~(\ref{eq:helix_sym}), with $\kappa$ and $\tau$ each held constant from site to site, is precisely that of a discrete \emph{circular} helix: the orbit of a single rigid screw motion, with constant radius and constant pitch. This differs from the weaker notion supplied by the classical theorem of Lancret~\cite{lancret1806memoire, do1976differential}, which states that a space curve is a \emph{generalized} helix if and only if the ratio $\kappa/\tau$ is constant. A generalized helix permits $\kappa$ and $\tau$ to vary in concert so long as their ratio is fixed; the circular helix requires each to be constant separately and is therefore strictly stronger. The vanishing of the integrability error selects the stronger object: $E[n]=0$ is the algebraic signature of a discrete circular helix, not of a Lancret generalized helix. Because $E[n]$ enters only through $\cos\tau[n]$, it fixes the curvature $\kappa$ and the torsion magnitude $|\tau|$, whereas the sign of $\tau$, the handedness of the screw orbit, remains an independent datum, precisely the chiral degree of freedom isolated in Barrier~I (Sec.~\ref{sec:barrier1}).

This identification gives $E[n]$ a geometric interpretation beyond that of a fitting residual: it serves as a \emph{local measure of discrete helical symmetry breaking}. At each residue, $E[n]$ quantifies the degree to which the local backbone geometry deviates from a perfect discrete circular helix, i.e., from a fixed point of the screw-motion group. In this picture, helical segments correspond to near-symmetric regions ($E[n]\approx 0$), while loops and strands correspond to regions of broken symmetry ($E[n]\gg 0$).
The circular helix is the differential-geometric realization of the uniform integrable segment analyzed dynamically above: the screw-motion orbit $\mathbf r[n+1]=M\,\mathbf r[n]$ is the same object as the translation-invariant Bloch mode of the transfer matrix $T$, whose conserved quasi-momentum coincides in magnitude with the torsion, $|q|=|\tau|$ [Eq.~(\ref{eq:dispersion})]. The geometric condition $E[n]=0$ and the dynamical condition of an exactly conserved quasi-momentum are thus two readings of one structure, so that $E[n]$ measures, equivalently, the local breaking of helical symmetry and the local non-conservation of the Bloch quasi-momentum.

\subsection{Integrable rigidity and the descriptive--predictive boundary}
\label{sec:rigidity}

The transfer-matrix reading of the uniform segment (Sec.~\ref{sec:piecewise},
Eq.~(\ref{eq:transfer})) does more than locate helical segments: it fixes how much
structural freedom such a segment retains. The argument rests on a single elementary
observation, that a uniform segment is a constant-coefficient linear recurrence, and
holds within the class of purely local, real-potential reductions for which the
dispersion relation Eq.~(\ref{eq:dispersion}) was derived. We use the term \emph{integrable segment}
here in its weakest sense, a segment that is translation-invariant and carries one
conserved quasi-momentum.

\begin{lemma}[Uniform-segment recurrence and conserved quasi-momentum]\label{lem:rigidity}
Let $\mathcal{S}=\{n_0,\dots,n_1\}$ be a segment on which the curvature ratios satisfy
$r^{\pm}[n]\approx 1$, the couplings are uniform $\beta^{\pm}_n\approx\beta$, and the
effective potential is real and constant, $V_{\mathrm{eff}}[n]\equiv V_{\mathrm{re}}$.
Then the DNLS equation~(\ref{eq:dnls}) reduces on $\mathcal{S}$ to the
constant-coefficient, second-order linear recurrence
$\psi[n+1]=a\,\psi[n]-\psi[n-1]$ with $a=2+V_{\mathrm{re}}/\beta$, governed by the
transfer matrix $T$ of Eq.~(\ref{eq:transfer}) with $\det T=1$ and Bloch quasi-momentum
$|q|=|\tau|$ fixed by Eq.~(\ref{eq:dispersion}). Because $T$ is independent of $n$, the
site shift is a symmetry and $q$ (equivalently the discrete Wronskian) is conserved along
$\mathcal{S}$.
\end{lemma}

\begin{proof}
Setting $\beta^{\pm}_n=\beta$ and $V_{\mathrm{eff}}[n]=V_{\mathrm{re}}$ in
Eq.~(\ref{eq:dnls}) turns the left-hand side into $\beta$ times the discrete second
difference $\psi[n+1]-2\psi[n]+\psi[n-1]$; equating to $V_{\mathrm{re}}\psi[n]$ and
dividing by $\beta>0$ gives the stated recurrence, whose transfer matrix, eigenvalues
$\lambda_\pm=e^{\pm iq}$, and dispersion $|q|=|\tau|$ are those established in
Eqs.~(\ref{eq:transfer})--(\ref{eq:dispersion}) above. Constancy of $a$ makes $T$ commute
with the site shift, conserving $q$.
\end{proof}

Here $V_{\mathrm{re}}$ is read from the known backbone geometry rather than solved for
self-consistently; the lemma is a kinematic statement about the potential profile of a
given segment, in keeping with the role of the Hasimoto map as a kinematic identity.

\begin{proposition}[Boundary rigidity]\label{prop:rigidity}
Under the hypotheses of Lemma~\ref{lem:rigidity}, the solution space on $\mathcal{S}$ is
exactly two-dimensional, and the two boundary data, namely the endpoint values, or
equivalently a single adjacent pair $(\psi[n_0],\psi[n_0-1])$ (an initial position,
tangent, and normal in the reconstruction of the lines preceding
Eq.~(\ref{eq:frame})), determine the entire interior $\{\psi[n]:n_0<n<n_1\}$ uniquely.
No interior degree of freedom remains.
\end{proposition}

\begin{proof}
Since $\det T=1\neq 0$, $T$ is invertible and the map from a solution to its seed
$(\psi[n_0],\psi[n_0-1])\in\mathbb{C}^2$ is a linear isomorphism: the seed determines
$\psi[n]=(T^{\,n-n_0}\,\text{seed})_1$ for every $n$, forward and backward. Hence the
solution space is two-dimensional and one adjacent-pair seed fixes the interior
uniquely. (The two-endpoint form is equivalent away from the measure-zero standing-wave
resonances $q(n_1-n_0)\in\pi\mathbb{Z}$.) At the band edges $q\in\{0,\pi\}$ the roots
coincide and the second solution is the secular mode $(A+Bn)(\pm1)^n$; the space is
still two-dimensional, so the conclusion is unchanged.
\end{proof}

\begin{corollary}[Description, not prediction]\label{cor:nogo}
Within the purely local, real-potential reduction, the interior conformation of an
integrable segment is fixed by $(\beta,V_{\mathrm{re}})$ together with the two boundary
data (Proposition~\ref{prop:rigidity}). The amino-acid sequence enters this interior only
through the bond stiffnesses $\beta^{\pm}_n=1/|\mathbf{t}_n|$ within the single scalar
$V_{\mathrm{re}}$; the dispersion relation $\cos\tau=1+V_{\mathrm{re}}/2\beta$ carries no
other sequence-dependent quantity, and this channel accounts for less than $5\%$ of the
variance of $V_{\mathrm{re}}$ on our dataset (Sec.~\ref{sec:decomposition}), the remainder
being set by the local geometry $(\kappa,\tau)$. The conserved quasi-momentum
$|q|=|\tau|$ is thus a quantity read off a known segment rather than generated from
sequence: an integrable reduction describes a backbone whose endpoints are given, but
does not predict an interior conformation from sequence. The map that \textit{de novo} prediction
would require, from sequence to a specific three-dimensional interior of dimension
$2(L-1)$ for a segment of length $L$, is precisely what the boundary rigidity removes.
\end{corollary}

\paragraph{Interpretation.}
A uniform segment is translation-invariant, carries a single conserved quasi-momentum,
and is fixed by its boundaries; this rigidity is what makes the segment integrable and,
at the same time, what prevents its interior from encoding a sequence-dependent fold. The
conclusion is a property of the local, real-potential reduction, not of the Hasimoto map,
and it constrains only the interior of a segment. The relative placement of
secondary-structure elements is a separate matter, governed by information the local
dispersion relation does not carry (Sec.~\ref{sec:piecewise}). Nor does it say that
sequence lacks the physico-chemical information needed to fold a protein: that information
is real, but the present model does not couple it into the quantities that fix the
interior. A predictive model must do exactly that, incorporating amino-acid information into its parameters. This
is consistent with the predictive successes of DNLS-based models, which supply the missing
conformational information externally, through dissipative relaxation dynamics (Monte
Carlo or Glauber, rather than integrable time evolution) together with injection of the
native structure via fitted parameters, native-seeded external
fields~\cite{liubimov2025modeling}, or an RMSD-to-native term in the cost
functional~\cite{chernodub2010topological}. In a genuine \textit{de novo} setting that injected
information is absent, and the boundary data the rigidity requires is itself the missing
answer.

\section{Self-consistent field test}
\label{sec:scf}

The two barriers identified in Sec.~\ref{sec:decomposition} are static: they concern
the information content of $V_{\text{eff}}$ extracted from known structures and are
properties of the local, real-potential reduction itself. They follow from discarding
$V_{\text{im}}$ and from the geometric dominance of $V_{\text{re}}$
[Eqs.~(\ref{eq:Vre})--(\ref{eq:Vim})], and hence persist for any choice of
energy-function parameterization within this model class. This point is
demonstrated directly by what we term an \emph{oracle test}: in
Fig.~\ref{fig:vim_info}(b) we supply the exact $V_{\text{re}}[n]$ computed from the
native structure, which constitutes the best possible real potential that any energy
function could produce, and attempt backbone reconstruction with $V_{\text{im}}=0$.
The resulting RMSD of 20--120\,\AA{} across all chain lengths (mean ${\sim}50$\,\AA; 40--120\,\AA{} for the longest 200--300-residue chains, cf.\ Sec.~\ref{sec:decomposition}) shows that even perfect
knowledge of $V_{\text{re}}$ is insufficient, because the $2^{N}$ torsion-sign
degeneracy is intrinsic to the real-potential reduction and cannot be resolved by any choice of
energy function.

A separate question is whether the DNLS can be used as a dynamical equation to drive
an initially unfolded chain toward the native state, given a physically motivated
effective potential. A natural dynamical hypothesis is that the DNLS could drive folding through modulation
instability, in which small perturbations of a uniform solution grow
exponentially and might seed the formation of soliton-like secondary-structure elements.
We test this dynamical scenario directly by performing self-consistent
field (SCF) iterations on all 856 proteins in our dataset, using two independent
settings: one with hydrophobic and elastic potentials only, and one with an additional
hydrogen-bond term. The SCF test provides a third, dynamical barrier that is
independent of Barriers~I and~II and serves as an additional confirmation of the
representational limitation.

\subsection{SCF formulation}

We initialize the complex field as a nearly uniform state
$\psi_{0}[n] = \kappa_{0} + \delta\psi[n]$, where $\kappa_{0} = 0.1$\,rad represents a
nearly straight chain and $\delta\psi[n]$ is a small random perturbation with
$|\delta\psi| < 0.01$. The field is evolved under damped dynamics
\begin{equation}
	\frac{d\psi[n]}{dt} = -(1 + i\gamma)\,\frac{\delta H[\psi]}{\delta\psi^{*}[n]}\,,
	\label{eq:scf_dynamics}
\end{equation}
where $\gamma = 0.5$ is a dissipation parameter that drives the system toward a local
energy minimum rather than conserving the Hamiltonian. The functional $H[\psi]$ consists
of a kinetic (elastic) term and an interaction term:
\begin{equation}
	H[\psi] = \sum_{n}\beta\,|\psi[n+1]-\psi[n]|^{2}
	+ \sum_{n}V_{\text{int}}[n]\,|\psi[n]|^{2}\,.
	\label{eq:scf_hamiltonian}
\end{equation}

The interaction potential $V_{\text{int}}$ is constructed from two contributions. The
elastic term $V_{\text{elastic}}[n] = \lambda(\kappa[n] - \kappa_{\text{target}})^{2}$
penalizes deviations of the local curvature from a target value
$\kappa_{\text{target}} = \langle\kappa\rangle_{\text{native}}$ computed from the known
structure. The hydrophobic term $V_{\text{hydro}}[n]$ is a non-local potential that
assigns a favorable energy to contacts between hydrophobic residues:
\begin{equation}
	V_{\text{hydro}}[n] = -\sum_{m,\,|m-n|>4} h[n]\,h[m]\,
	f\!\bigl(|\mathbf{r}[n]-\mathbf{r}[m]|\bigr)\,,
	\label{eq:vhydro}
\end{equation}
where $h[n]\in\{0,1\}$ is the hydrophobicity index of residue $n$ (assigned according to
the Kyte-Doolittle~\cite{kyte1982simple} scale with a binary cutoff) and $f(r)$ is a contact function that
equals unity for $r < 8$\,\AA\ and decays smoothly to zero beyond this distance. The
backbone coordinates $\mathbf{r}[n]$ are reconstructed from $\psi[n]$ at each SCF step
via the inverse Hasimoto transform [Eq.~(\ref{eq:inverse})] and discrete Frenet
reconstruction.

In the second setting we add a hydrogen-bond potential
\begin{equation}
	V_{\text{hb}}[n] = -\epsilon_{\text{hb}}\sum_{m=n+3}^{n+5}
	g\!\bigl(\kappa[n],\kappa[m]\bigr)\,
	f\!\bigl(|\mathbf{r}[n]-\mathbf{r}[m]|\bigr)\,,
	\label{eq:vhb}
\end{equation}
where $\epsilon_{\text{hb}}$ is a coupling strength and $g(\kappa[n],\kappa[m])$ is a
geometric filter that favors the curvature values characteristic of helical hydrogen bonds.
This term is designed to mimic the $i \to i+4$ backbone hydrogen bond that stabilizes
$\alpha$-helices, expressed entirely in terms of the $\psi$ field variables.

The SCF iteration proceeds as follows. Starting from $\psi_{0}$, we (i) reconstruct the
backbone coordinates, (ii) evaluate $V_{\text{int}}$ (and $V_{\text{hb}}$ in the second
setting), (iii) update $\psi$ according to Eq.~(\ref{eq:scf_dynamics}) with a discrete
time step $\Delta t = 0.01$, and (iv) repeat for 5000 steps. At each step we record the
mean curvature $\langle\kappa\rangle = N^{-1}\sum_{n}|\psi[n]|$ and the C$_\alpha$ RMSD
between the reconstructed backbone and the native structure (after optimal
superposition). The best RMSD achieved over the trajectory is reported as the outcome
for each protein.

\subsection{Results without hydrogen bonds}

\begin{figure*}[t]
	\centering
	\includegraphics[width=0.8\textwidth]{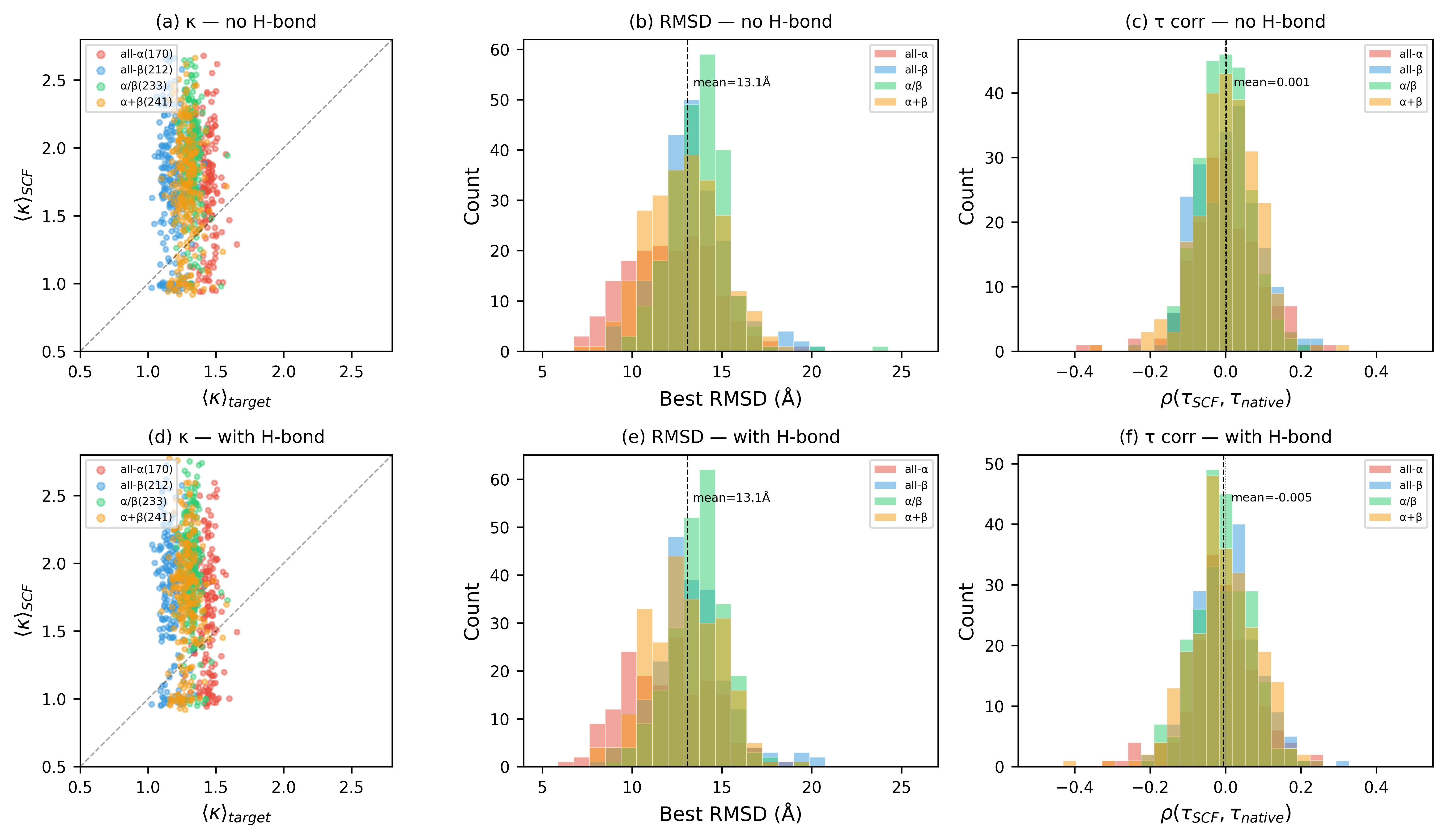}
	\caption{Self-consistent field (SCF) test of DNLS-driven folding on 856 non-redundant
		proteins, without (top row) and with (bottom row) a hydrogen-bond potential. Points
		and histograms are colored by SCOP class: all-$\alpha$ (170, red), all-$\beta$
		(212, blue), $\alpha/\beta$ (233, green), $\alpha$+$\beta$ (241, orange).
		(a,\,d) Mean bond curvature $\langle\kappa\rangle_{\text{SCF}}$ vs.\ native target
		$\langle\kappa\rangle_{\text{target}}$; the dashed line marks perfect agreement. In
		both settings $\langle\kappa\rangle_{\text{SCF}}$ is systematically overestimated
		(mean 1.776 and 1.790 vs.\ target 1.298), with no SCOP class approaching the
		diagonal. (b,\,e) Distribution of best-achieved RMSD. The two histograms are
		statistically indistinguishable (mean 13.1\,\AA; range 6.8--24.0\,\AA\ without
		hydrogen bonds, 6.7--26.4\,\AA\ with hydrogen bonds), and no protein in either
		setting reaches RMSD $< 5$\,\AA. (c,\,f) Pearson correlation
		$\rho(\tau_{\text{SCF}}, \tau_{\text{native}})$ between the SCF-predicted and native
		torsion-angle profiles. Both distributions are centered on zero (mean 0.001 vs.\
		$-0.005$), indicating that the SCF torsion angles are uncorrelated with the native
		structure. The near-identical results across the two rows demonstrate that the
		folding failure is not due to missing hydrogen-bond physics but to the
		real-potential SCF dynamics on the single complex scalar field $\psi$, whose dissipative evolution couples amplitude and phase and so cannot
		independently regulate the curvature and torsion constraints required for structure
		determination.}
	\label{fig:scf}
\end{figure*}

Figure~\ref{fig:scf}(a) shows the final mean curvature
$\langle\kappa\rangle_{\text{SCF}}$ plotted against the native target
$\langle\kappa\rangle_{\text{target}}$ for all 856 proteins, colored by SCOP class. The
dashed diagonal marks perfect agreement. The SCF systematically overestimates the mean
curvature: the dataset-wide mean of $\langle\kappa\rangle_{\text{SCF}}$ is 1.776\,rad
compared to a target of 1.298\,rad, an overestimation of approximately 37\%. No SCOP
class approaches the diagonal. The overestimation reflects the tendency of the damped
DNLS dynamics to produce excessive bending: the hydrophobic potential drives chain
compaction, which in the $\psi$ representation manifests as increased $|\psi| = \kappa$,
but without the directional constraints of hydrogen bonds there is no mechanism to
regulate the curvature to its native value.

Figure~\ref{fig:scf}(b) displays the distribution of best-achieved RMSD across the 856
proteins. The mean RMSD is 13.1\,\AA\ with a range of 6.8 to 24.0\,\AA. No protein
achieves RMSD below 5\,\AA, a threshold commonly used to indicate a successful fold
prediction. The distribution shows no significant separation among the four SCOP classes:
all-$\alpha$, all-$\beta$, $\alpha/\beta$, and $\alpha$+$\beta$ proteins are
interleaved throughout the histogram. This uniformity confirms that the SCF failure is
not specific to any particular fold topology.

Figure~\ref{fig:scf}(c) shows the Pearson correlation $\rho(\tau_{\text{SCF}},
\tau_{\text{native}})$ between the torsion-angle profile produced by the SCF and the
native torsion-angle profile, computed for each protein. The distribution is centered on
zero (mean $\rho = 0.001$, median $\rho = 0.000$) and extends symmetrically over the
range $[-0.4, 0.4]$. The SCF-generated torsion angles therefore bear no
systematic relationship to the native structure.

The absence of correlation between predicted and native torsion angles constitutes a clear failure of the \textit{de novo} predictive scheme. Although the SCF dynamics may induce chain compaction, reflected in the curvature profile, the torsion angles $\tau[n]$, which govern the global fold topology, remain effectively stochastic. The result is a compact globule with no native structural fidelity: the isotropic hydrophobic
potential drives compaction (increased $\kappa$) but cannot select the specific
torsion-angle sequence that defines the native fold. This is consistent with the established distinction between
hydrophobic collapse and folding. The selection of specific secondary and tertiary structure elements involves directional
hydrogen bonds, whose geometric constraints are difficult to represent within the scalar DNLS framework.

\subsection{Results with hydrogen bonds}

The bottom row of Fig.~\ref{fig:scf} shows the same three diagnostics for the SCF with
the hydrogen-bond potential [Eq.~(\ref{eq:vhb})] included. The results are statistically
indistinguishable from the setting without hydrogen bonds. The mean curvature increases
slightly to $\langle\kappa\rangle_{\text{SCF}} = 1.790$\,rad (compared to 1.776 without
the hydrogen-bond term), moving further from the target rather than closer. The RMSD
distribution is virtually identical: mean 13.1\,\AA, range 6.7 to 26.4\,\AA, with no
protein reaching RMSD $< 5$\,\AA. The torsion correlation shifts negligibly from
$\rho = 0.001$ to $\rho = -0.005$, remaining consistent with zero.

The universal failure of the hydrogen-bond term to improve the SCF outcome demonstrates that the barrier to folding within the DNLS framework is not merely energetic, but representational: the single complex scalar field $\psi[n] = \kappa[n]\exp(i\sum\tau[k])$ entangles curvature and torsion in its amplitude and phase. A physical backbone hydrogen bond imposes strictly independent constraints on spatial distance (governed by $\kappa$) and directional orientation (governed by $\tau$) between specific residue pairs. The evolution of $\psi$ couples its amplitude $|\psi|$ and phase $\arg(\psi)$ through the nonlinear dynamics, which prevents the independent regulation of curvature and torsion required to stabilize specific donor-acceptor geometries. The hydrogen-bond potential [Eq.~(\ref{eq:vhb})], expressed in terms of $\psi$, therefore cannot enforce the discrete orientational constraints that the physical interaction requires.

A natural question is whether replacing the interaction potential
$V_{\text{int}}$ with a more refined functional, for instance a
C$_\alpha$-level effective potential obtained by systematic coarse-graining
of an all-atom force field such as AMBER~\cite{cornell1995second, case2005amber} or CHARMM~\cite{brooks1983charmm, brooks2009charmm}, would alter this
conclusion. Three independent lines of evidence indicate that it would not.
First, the oracle test [Fig.~\ref{fig:vim_info}(b)] already supplies the
exact native $V_{\text{re}}[n]$, which represents the theoretical optimum
for the real part of the potential; the resulting RMSD of 20--120\,\AA{}
demonstrates that the $2^{N}$ torsion-sign degeneracy (Barrier~I) persists
irrespective of energy-function quality. Second, the geometric dominance
result (Barrier~II, $\rho_{\text{geom}}=0.951$) establishes that
$V_{\text{re}}$ is 95\% determined by backbone geometry rather than by
amino-acid identity, so that even an exact force-field projection would
produce a $V_{\text{re}}$ profile nearly indistinguishable from the
geometry-only version. Third, the representational bottleneck is structural
rather than parametric: as described above, the Hasimoto field $\psi$ encodes the two real degrees
of freedom $\kappa$ and $\tau$ in a single complex number whose amplitude
and phase remain coupled under the dynamics, so that no force field, however
accurately parameterized, can independently enforce the distance and angular constraints
that directional interactions impose on specific residue pairs.

The SCF test with its deliberately minimal energy function
therefore probes the representational ceiling of the $\psi$ framework, not
the quality of a particular force field.

This conclusion extends to state-of-the-art machine-learning force fields (MLFFs) such as AI2BMD~\cite{wang2024ab}, MACE~\cite{batatia2022mace}, and ViSNet~\cite{wang2024enhancing}. These architectures achieve near-quantum accuracy by operating on explicit Cartesian coordinates, but the three arguments above apply to them unchanged: the bottleneck is representational, not energetic, so higher chemical fidelity cannot overcome the limitation of the Hasimoto framework. Because $V_{\text{re}}$ is determined ${\sim}95\%$ by geometry rather than sequence (Barrier~II), the high-fidelity chemical information provided by MLFFs is largely filtered out when projected onto the DNLS potential. The real-potential reduction, which retains only $V_{\text{re}}$ and discards the parity-odd $V_{\text{im}}$, thus acts as a lossy projection that drops the chiral degrees of freedom needed to distinguish the native state, a structural deficit that no improvement in force-field accuracy can overcome. This same limitation is recognized empirically in the MLFF literature: recent work reports that most existing potentials assume localized interactions and that explicitly modeling \emph{non-local} electrostatic and multipole contributions reduces energy and force errors by more than 50\% in biomolecular simulations~\cite{cui2026enhancing}. Our parity--locality classification (Table~\ref{tab:parity_locality}) offers a first-principles account of why such non-local terms are indispensable: the interactions that select tertiary structure populate the non-local row that any nearest-neighbor scalar potential, including the DNLS, cannot represent.

\subsection{Summary of the dynamical barrier}

The SCF experiments establish a third, dynamical barrier to using the local,
real-potential DNLS reduction for protein structure prediction. The barrier can be stated concisely: the dissipative DNLS
dynamics, even supplemented with physically motivated interaction terms, cannot drive an
unfolded chain to the native state, because the single complex scalar field $\psi$ lacks
the representational capacity to encode the independent curvature and torsion constraints
imposed by directional hydrogen bonds.
Three quantitative findings support this conclusion across all 856 proteins and both SCF
settings. First, the mean curvature is systematically overestimated by approximately 37\%,
indicating that the dynamics produces excessive bending without the regulatory effect of
hydrogen-bond geometry. Second, no protein achieves RMSD below 5\,\AA, and the RMSD
distribution shows no dependence on SCOP class, confirming that the failure is universal
rather than fold-specific. Third, the torsion correlation with the native structure is
indistinguishable from zero ($\rho \approx 0$), demonstrating that the SCF dynamics
generates effectively random torsion-angle sequences.
Comparing the two SCF settings isolates the cause of the failure. The addition of
a hydrogen-bond potential within the DNLS framework produces no measurable improvement in
any of the three diagnostics. This rules out the interpretation that the SCF failure is
due to an incomplete energy function. The failure is instead structural, originating in the $\kappa$--$\tau$ entanglement of $\psi$ described above. This \textit{de novo} failure does not contradict the established
predictive successes of DNLS-based models: when the free-energy parameters are fit to a
known native structure, Monte Carlo (or Glauber) relaxation can indeed recover that
structure from a random-walk phase~\cite{begun2021topology,liubimov2025modeling}. Our
claim is restricted to the \textit{de novo} setting, in which no native structure is available to
be injected, whether into the parameters, an external field, or an RMSD-to-native cost, and the
dissipative dynamics is driven only by sequence-level potentials. The relationship between
this dynamical barrier and the two static barriers
of Sec.~\ref{sec:decomposition} is discussed in Sec.~\ref{sec:discussion}.

\section{Discussion and Conclusion}
\label{sec:discussion}

The Hasimoto transform was originally constructed for vortex filaments
evolving under the local induction approximation (LIA). In that setting
the transform is both kinematic and dynamic: it converts the geometric
evolution of the filament into the integrable cubic NLS, whose soliton
solutions describe physical excitations that propagate without
dispersion. The success of this construction rests on four properties
of the vortex system that proteins do not share: locality of
interactions, homogeneity of the medium, Hamiltonian (non-dissipative)
dynamics, and a physically determined effective potential. Protein
folding violates all four. Hydrophobic contacts and electrostatic
forces are non-local; the backbone comprises 20 chemically distinct
amino acids; folding is a dissipative free-energy minimization in
aqueous solvent; and, as shown in Sec.~\ref{sec:decomposition}, the
effective potential $V_{\text{eff}}$ on proteins is an algebraic
consequence of the backbone geometry rather than an independently
specified physical quantity. Table~\ref{tab:comparison} summarizes these distinctions.
Each of the four mismatches removes a degree of freedom that the predictive problem requires: non-locality, chemical heterogeneity, dissipation, and a sequence-specified potential are precisely the ingredients that a local, homogeneous, Hamiltonian, geometry-fixed reduction leaves out. The reduction is in this sense kinematically complete but dynamically inert: it faithfully re-expresses a given geometry, yet retains none of the sequence-specific, non-local information from which that geometry would have to be generated.

\begin{table}[t]
	\centering
	\caption{Comparison of the Hasimoto map applied to vortex filaments
		and to protein backbones. The four properties listed are necessary
		for the map to function as a predictive dynamical framework.}
	\label{tab:comparison}
	\small
	\setlength{\tabcolsep}{4pt}
	\begin{tabular}{lcc}
			\toprule
			Property & Vortex filament & Protein backbone \\
			\midrule
			Interactions & Local (LIA) & Non-local \\
			Medium & Homogeneous & Heterogeneous (20 AA) \\
			Dynamics & Hamiltonian & Dissipative \\
			$V_{\text{eff}}$ & Physical potential & Kinematic identity \\
			\bottomrule
		\end{tabular}
\end{table}

The geometric formalism established by Niemi and collaborators provides a rigorous basis for the structural characterization of protein backbones. The identification of discrete soliton solutions with secondary-structure motifs reveals the geometric regularity underlying secondary structure, while multi-soliton ans\"atze yield compact parameterizations of folded states with sub-\aa ngstr\"om accuracy~\cite{molkenthin2011discrete,krokhotin2012soliton}. Recent extensions have successfully applied this basis to simulate thermal unfolding~\cite{begun2021topology} and characterize topological phase transitions~\cite{begun2025local}. These studies constitute a comprehensive treatment of the \emph{inverse problem}: given a known topology, the DNLS soliton basis affords an efficient representation and a robust framework for analyzing fluctuations around the native state.
The topological aspects of this program connect to a broader body of work on protein topology. Approximately 6\% of known protein structures form knots, slipknots, or links whose folding requires the backbone to cross topological barriers~\cite{sulkowska2020folding, dabrowski2017topological}. These entangled proteins highlight a fundamental limitation of any local geometric description: the topological invariants that distinguish a knotted from an unknotted fold are inherently global properties that cannot be fully determined from the local fields $(\kappa[n],\tau[n])$ at any finite number of sites. Consequently, soliton-based studies typically determine parameters by fitting to a known crystallographic structure, whether through Metropolis minimization of RMSD~\cite{molkenthin2011discrete} or direct fitting to native coordinates~\cite{begun2021topology}.

That a known fold can be represented so faithfully has a precise structural explanation in the integrability
hierarchy of discrete nonlinear equations. The geometric free energies employed in this
program are, after the Hasimoto map, \emph{on-site} discrete nonlinear Schr\"odinger
functionals: their nonlinear terms (e.g.\ the quartic double well
$\lambda(\kappa_i^2-m^2)^2$ and the coupling $\kappa_i^2\tau_i^2$) are diagonal in the
site index, the only inter-site term being the linear hopping
$\kappa_{i+1}\kappa_i$~\cite{begun2025local}. Unlike the Ablowitz--Ladik lattice~\cite{ablowitz1976nonlinear}, the
integrable discretization of NLS, which carries a Lax pair and an infinite hierarchy of
conserved quantities, the on-site discrete NLS is non-integrable: it possesses no Lax
pair and no conserved hierarchy beyond the global phase, and its uniform-segment limit
retains only the single quasi-momentum of the constant-coefficient transfer
matrix~(\ref{eq:transfer}) (the discrete probability current of
Eq.~(\ref{eq:current}) is the conserved bilinear of the \emph{linear} eigenproblem,
distinct from the integrable \emph{hierarchy} absent here). The $\alpha$-helices are therefore isolated segments of
\emph{piecewise linear integrability}, not manifestations of a globally integrable
dynamics. This is the structural origin of the descriptive--predictive asymmetry: a
fully integrable system is rigidly constrained by its conserved hierarchy, so initial
local data propagate to a unique global trajectory; the on-site model has no such
rigidity, and its soliton solutions are obtained by direct energy minimization against a
known target rather than by integrating the sequence forward. This structural feature also motivates the approach we take throughout. Rather than following the usual integrable-systems route of solving the equation for explicit soliton solutions and examining their stability, we take the integrable state itself as a fixed kinematic reference and measure how real backbones depart from it, using the exact decomposition, the integrability error $E[n]$, and the parity classification. Because the protein potential is reconstructed from a known geometry, solving for solitons would merely return the structure already supplied; the productive direction is the reverse one, reading the integrable ideal as a ruler for the deviations that carry the folding information. The absence of a Lax pair
is thus not a technical detail but the reason the framework can represent a known fold
yet cannot generate an unknown one from local sequence data alone.

In contrast, the present work interrogates the complementary \textit{forward problem}: determining whether the DNLS framework permits \textit{ab initio} prediction of the native structure solely from the amino-acid sequence. Our analysis indicates that structural barriers impede this predictive pathway. The exact decomposition (Sec.~\ref{sec:decomposition}) reveals that the effective potential $V_{\text{eff}}$ is determined predominantly by the target geometry rather than the sequence, creating an informational circularity that cannot be resolved by sequence-based potentials. The SCF experiments (Sec.~\ref{sec:scf}) further show that physically motivated interactions fail to lift the $2^{N}$ torsion-sign degeneracy or overcome the representational limitations of the scalar field. These findings indicate that the distinction between the inverse and forward problems is structural: the capacity to describe known folds does not imply the feasibility of predicting unknown ones. The specificity that prediction demands is precisely the non-local, sequence-dependent content discarded in passing to a local real-potential model, and rearranging that model does not restore it.
One might nonetheless hope to supply that missing specificity dynamically. The hypothesis that modulation instability of the DNLS may provide a
dynamical mechanism for the emergence of secondary structure from an
initially straight chain is directly tested by our SCF
experiments. The instability does produce curvature
growth from a nearly uniform initial state, consistent with the
expected behavior. However, the resulting structures bear no
resemblance to native folds: the torsion angles are uncorrelated with
the native profile ($\rho \approx 0$), and the RMSD remains above
5\,\AA\ for all 856 proteins in both SCF settings. Modulation
instability generates generic bending but cannot select the specific
$(\kappa,\tau)$ sequence that defines a particular protein fold.

Although this on-site, real-potential reduction cannot serve as a predictive tool, the exact
decomposition and the analyses built upon it yield three constructive
results. First, the integrability error $E[n]$
[Eq.~(\ref{eq:E_def})] serves as a geometric probe of secondary
structure. The ROC analysis (Fig.~\ref{fig:roc}, AUC $= 0.72$)
demonstrates that 72\% of the helix/non-helix distinction defined by
DSSP is captured by a scalar test of discrete helical symmetry applied
to C$_\alpha$ coordinates alone, without reference to hydrogen-bond
energies or side-chain identities. Whether this geometric criterion
can provide useful secondary-structure annotation in data-limited
settings, such as low-resolution density maps or coarse-grained
trajectories, remains to be tested on independent benchmarks; the
present AUC establishes the conceptual correspondence between DNLS
integrability and helical geometry but does not by itself validate a
practical tool. The gap between AUC $= 0.72$ and unity quantifies the
structural information that hydrogen bonds and non-local interactions
contribute beyond local geometric regularity. Second, the effective potential
$V_{\text{eff}}[n]$ provides a structural fingerprint that is invariant
under rigid-body transformations. Because $V_{\text{re}}$ is 95\%
determined by geometry (Sec.~\ref{sec:barrier2}), homologous proteins with low
sequence identity but similar folds produce similar $V_{\text{re}}[n]$
profiles, offering a one-dimensional scalar representation of backbone
shape for structure comparison without spatial superposition. Third, the
information-theoretic lower bound on chiral information loss is exact
within the purely local, real-potential model class: each residue contributes one bit of torsion-sign
information encoded in $V_{\text{im}}$ but absent from $V_{\text{re}}$,
producing a $2^{N}$-fold degeneracy that any predictive scheme based on
this class of models would need to resolve.

Comparison with the data-driven methods that have solved
the prediction problem in practice throws these structural barriers into sharper
relief. AlphaFold~2~\cite{jumper2021highly} and ESMFold~\cite{lin2023evolutionary} both predict the full $SE(3)$ rigid-body frame, comprising a rotation matrix and a translation vector, for every residue. This approach retains the complete orientational degrees of freedom of the discrete Frenet frame without projecting them onto a scalar field. The theoretical foundations of this representational choice are clarified by the geometric deep learning framework~\cite{bronstein2021geometric, bronstein2017geometric}, which demonstrates that neural architectures respecting the symmetry group of the data domain (specifically, $SE(3)$ equivariance for molecular geometry) achieve systematic gains in sample efficiency and
generalization by building physical invariances directly into the
network structure rather than learning them from data. The full-frame
representation retains both signs of the torsion explicitly, whereas the
local, real-potential reduction of the DNLS, which keeps only
$V_{\text{re}}=\operatorname{Re}V_{\text{eff}}$, projects onto the
parity-even sector, and our Barrier~I shows that this projection
discards the torsion sign, introducing a $2^{N}$ chiral
degeneracy that the full-frame representation avoids entirely.
Barrier~II reveals a second contrast: the DNLS effective potential
is 95\% determined by local backbone geometry, whereas the attention mechanism of AlphaFold and the language-model embeddings of ESMFold encode
long-range coevolutionary and contextual information that couples
distant residues, information that statistical-physics methods such
as direct coupling analysis~\cite{morcos2011direct, cocco2013principal}
have shown to be extractable from evolutionary data via the
maximum-entropy principle. The DNLS framework, operating on nearest-neighbor
differences of $\psi$, has no analogous channel for non-local
sequence information. These comparisons do not diminish the value of the geometric approach, since data-driven models supply predictions rather than physical explanations. They do, however, delineate the specific representational deficits that any future analytical theory must overcome: it must preserve the full $SE(3)$ frame at each residue
and incorporate non-local, sequence-dependent interactions that go
beyond the scalar Hasimoto field.

The three barriers identified in this work point to specific physical
ingredients that a geometric theory of protein folding would need to
incorporate. Hydrogen bonds impose simultaneous constraints on the
distance and relative orientation of donor and acceptor groups, coupling
$\kappa$ and $\tau$ at specific residue pairs in a manner that cannot be
captured by a potential acting on the single complex field $\psi$. A
geometric energy functional treating $\kappa[n]$ and $\tau[n]$ as
independent fields, rather than combining them into a single complex
scalar, would be a natural starting point. The hydrophobic effect, an
entropic force mediated by solvent reorganization, is inherently
non-local and temperature-dependent, with no natural representation in
the DNLS Hamiltonian. The SCF experiments confirm that a contact-based
hydrophobic potential produces chain compaction but cannot select
secondary structure, consistent with the known distinction between
hydrophobic collapse and folding. Finally, the geometric dominance of
$V_{\text{re}}$ means that sequence-dependent terms going beyond the
weak modulation of virtual-bond lengths in $\beta^{\pm}_{n}$ would need
to enter any predictive geometric framework explicitly.

In summary, we have derived an exact closed-form decomposition of the
DNLS effective potential on protein backbones and used it to identify
three structural barriers to forward structure prediction: the
torsion-sign degeneracy encoded in $V_{\text{im}}$, the geometric
dominance of $V_{\text{re}}$, and the universal failure of
self-consistent field dynamics across 856 non-redundant proteins. We
further showed, in dynamical-systems terms, that the rigidity of an
integrable uniform segment is what removes the freedom
prediction requires, since boundary data fix the interior. A
parity-by-locality classification of the folding forces then explains why the
tertiary-structure drivers lie outside the local, real-potential sector.
These
barriers are mathematical in nature and are unlikely to be resolved by
parameter adjustment or algorithmic improvement within the local, real-potential DNLS
reduction. The Hasimoto map applied to proteins functions primarily as a kinematic identity
rather than a dynamical equation: it provides a useful geometric
language for describing folded states, a purely geometric helix
detector, and a rotation-invariant structural fingerprint, but within the scope of our analysis the local, real-potential reduction built on it does
not provide a viable pathway to \textit{ab initio} protein structure
prediction. As the comparison with data-driven methods above makes concrete, a future analytical theory of folding must go beyond the scalar Hasimoto field to preserve the geometric and chemical degrees of freedom that the folding process requires. The piecewise picture that emerges here, in which the integrability error partitions the backbone into near-integrable segments where $E[n]\to0$ and the broken-integrability regions between them, suggests a natural question for future work: whether the boundary between the ordered, soliton-like helical phase and the disordered coil phase admits a sharp theoretical characterization, in the sense of a limiting form for the order parameter across the transition. The descriptive--predictive tension we locate here for the protein backbone may, more generally, be characteristic of integrable reductions of heterogeneous physical systems.

\section*{Acknowledgments}
We thank Prof. Yan-Hong Qin (Xinjiang University) for inspiration and foundational guidance that motivated this work. This work was carried out in part using computing resources at the Computing and Data Center of Xinjiang University.

\section*{Declaration of competing interest}
The authors declare that they have no known competing financial interests or personal relationships that could have appeared to influence the work reported in this paper.

\section*{CRediT authorship contribution statement}
\textbf{Yiquan Wang:} Conceptualization, Methodology, Software, Validation, Formal analysis, Investigation, Data curation, Writing -- original draft, Writing -- review \& editing, Visualization.

\section*{Data and Code Availability}
The source code for this study is available on GitHub at \url{https://github.com/wyqmath/discrete_hasimoto_protein}.

\bibliographystyle{unsrt}
\bibliography{manuscript}
\end{document}